\documentclass{article}
\usepackage{wright}
\usepackage{tabularx}
\usepackage[table]{xcolor}
\pdfoutput=1

\AtEveryBibitem{
    \clearfield{day}
    \clearfield{month}
    \clearfield{series}
    \clearfield{venue}
    \clearname{editor}
    \clearlist{publisher}
    \clearlist{location}
    \clearfield{venue}
    \clearfield{issn}
    \clearfield{isbn}
    \clearfield{urldate}
    \clearfield{eventdate}
    \clearfield{pages}
    \clearfield{number}
    \clearfield{volume}
}

\renewcommand{\P}{\class{P}}
\newcommand{\NC}{\class{NC}}
\newcommand{\NP}{\class{NP}}
\newcommand{\BPP}{\class{BPP}}
\newcommand{\BQP}{\class{BQP}}
\newcommand{\AM}{\class{AM}}
\newcommand{\PSPACE}{\class{PSPACE}}
\newcommand{\EXP}{\class{EXP}}
\newcommand{\NEXP}{\class{NEXP}}
\newcommand{\FNEXP}{\class{FNEXP}}
\newcommand{\ALL}{\class{ALL}}
\newcommand{\BQPSPACE}{\class{BQPSPACE}}
\newcommand{\BQuPSPACE}{\class{BQ_UPSPACE}}

\newcommand{\puP}{\class{projUnitaryP}}
\newcommand{\puPSPACE}{\class{projUnitaryPSPACE}}
\newcommand{\uP}{\class{unitaryP}}
\newcommand{\uPSPACE}{\class{unitaryPSPACE}}
\newcommand{\uALL}{\class{unitaryALL}}

\newcommand{\Apx}{\normalfont\textsc{Apx}}
\newcommand{\Phase}{\normalfont\textsc{Phase}}
\newcommand{\Entry}{\normalfont\textsc{Entry}}
\newcommand{\Column}{\normalfont\textsc{Column}}

\title{Unitary complexity in polynomial space}
\author{William Kretschmer\thanks{University of Texas at Austin. Email: \texttt{kretsch@cs.utexas.edu}}
\and
Ewin Tang\thanks{Princeton University. Email: \texttt{ent@cs.princeton.edu}}}
\date{}

\begin{document}

\maketitle
\begin{abstract}
We show that if quantum commitments exist, then either there is no polynomial-time solution to the unitary synthesis problem, or $\mathsf{BPP} \neq \mathsf{NEXP}$.
Thus, showing unconditionally that quantum commitments exist would require answering at least one of two longstanding open questions in complexity theory.
We prove our main result as a consequence of a more general lemma, which shows that every unitary in $\mathsf{unitaryPSPACE}$ either cannot be synthesized efficiently relative to any classical oracle, or can be synthesized efficiently with an oracle for $\mathsf{NEXP}$ search problems.
Our lemma has other noteworthy consequences, including that certain oracle separations involving $\mathsf{unitaryPSPACE}$ would imply breakthrough classical lower bounds such as $\mathsf{NC} \neq \mathsf{NP}$.

Along the way, we propose new definitions for the unitary complexity classes $\mathsf{unitaryP}$ and $\mathsf{unitaryPSPACE}$.
Our changes address the biggest conceptual issues with definitions suggested in prior work, and lead to elegant proofs.
We study both implementations that erase garbage and implementations that allow it, because we cannot rule out the possibility that the two definitions differ.
Nevertheless, we show that both definitions can be viewed as special cases of each other.
We also showcase many other ways in which our definitions are robust.
For example, we show that $\mathsf{unitaryPSPACE}$ has an equivalent characterization as the set of unitary transformations whose entries can be computed to arbitrary precision in polynomial space.
Consequently, we deduce that $\mathsf{unitaryPSPACE}$ can generically erase garbage, a result that provably fails relative to unitary oracles.
\end{abstract}
\hypersetup{linktocpage}

\clearpage
\tableofcontents
\clearpage

\section{Introduction}

The complexity of quantum-input computational tasks may behave radically differently than the intuitions we have from classical computation.
For example, a basic principle in classical cryptography is that all superpolynomially-computationally secure cryptosystems rely on the hardness assumption $\P \neq \NP$, with few exceptions.
Yet, several works~\cite{Kre21-pseudorandom,KQST23-prs,KQT24-owf} have established that the same might not be true of \textit{quantum} cryptography, by constructing oracles relative to which $\NP$-complete problems are easy and yet substantial quantum cryptography exists.

Motivated by these separations, a flurry of work has attempted to characterize the world of cryptography that could exist without assuming $\P \neq \NP$, now regularly called ``Microcrypt''~\cite{Sat26-zoo}.
A large fraction of this effort has gone toward identifying the \textit{minimal} primitive that is necessary for quantum cryptography, much the same way that one-way functions are considered the foundation without which classical cryptography could not exist~\cite{IL89-owf,Imp95-average}.
Quantum commitment schemes have emerged as the most promising such candidate, as they have many equivalent notions and are implied by almost all other forms of quantum cryptography~\cite{Yan22-commitments,BCQ23-efi,HMY23-superpositions,KT24-commitments,BJ24-commitments,CCC+26-meta}.

The complexity-theoretic picture surrounding quantum commitments is particularly perplexing: commitments are not known to imply \textit{any} kind of hardness or separation in standard complexity theory.
In fact, some researchers are optimistic that the existence of quantum commitments might not be subject to the same barriers of relativization~\cite{BGS75-p-np}, natural proofs~\cite{RR97-nat-proofs}, or algebrization~\cite{AW09-algebrization} that have thus far prevented attacks on the $\P \stackrel{?}{=} \NP$ problem.
It is conceivable that we could unconditionally prove the computational security of a quantum commitment scheme, using only the complexity-theoretic lower bound techniques that we know of today.
Indeed, \cite{Qia24-unconditional,MNY24-unconditional} recently demonstrated this possibility in the \textit{auxiliary input model}, in which all parties receive copies of a (possibly exponentially-complex) advice state on a polynomial number of qubits.

So, does traditional complexity theory have \textit{nothing} to say about the hardness that underlies quantum cryptography?
Or, is there some yet-unknown method to show that quantum cryptography would still imply breakthrough classical lower bounds?
This uncertainty is best captured in the unresolved \textit{unitary synthesis problem}~\cite{AK07-qcma-qma,Ros21-unitary,LMW24-synthesis}.
First posed over 20 years ago, it asks whether unitary computation always reduces efficiently to classical computation.
More formally, a positive resolution to unitary synthesis means the existence of a $\poly(n)$-time quantum oracle algorithm such that for every $n$-qubit unitary transformation $U$, there exists a choice of the classical oracle function that allows the algorithm to approximately implement $U$.
An efficient unitary synthesis algorithm, then, would at least show that any superpolynomial lower bounds for a unitary implementation problem would yield superpolynomial hardness of \textit{some} classical function.

Still, it is not clear whether unitary synthesis ought to directly inform quantum cryptography or vice versa.
For one, the unitary synthesis problem as originally stated is purely existential; it says nothing about how the oracle function's complexity relates to the unitary $U$!
Conversely, even a negative solution to unitary synthesis could be highly nonexplicit.
A lower bound for synthesizing Haar-random unitaries $U$ might not be directly relevant to the hardness of any particular quantum commitment scheme.

The hope of using quantum states as a new foundation for cryptography and the lack of a clear connection to classical complexity theory have motivated the introduction of a ``quantum-native'' language of complexity theory that treats quantum transformations as first-class objects.
Whereas ordinary complexity classes such as $\P$, $\BQP$, or $\PSPACE$ capture classical-input decision problems, some newly-defined classes express \textit{unitary transformation} problems, where the computational task is to coherently apply a quantum operation on an unknown and arbitrary input state~\cite{RY22,MY23,BEMPQY26-uhlmann,Zha25-unitary-oracles}.
Two obvious choices for classes to study are $\uP$ and $\uPSPACE$, which correspond to the unitary transformations that can be implemented in polynomial time or space, respectively, on a quantum computer.
We defer a formal treatment of the definitions until \Cref{sec:detour}.
For now, we note that the $\uP \stackrel{?}{=} \uPSPACE$ problem has already emerged as a central question in the context of cryptography, because it has been shown that the existence of any falsifiable computationally-secure quantum cryptosystem implies the separation $\uP \neq \uPSPACE$~\cite{BEMPQY26-uhlmann}.
Thus, while we do not know whether the existence of quantum commitments would (say) imply a separation of $\P$ and $\PSPACE$, they would at least separate $\uP$ and $\uPSPACE$.

\subsection{This work}
\label{sec:this_work}

Broadly speaking, the first goal of this paper is to probe the computational hardness of $\uPSPACE$ using classical complexity-theoretic techniques.
For example, our most easily-stated result establishes a direct connection between quantum cryptography, the unitary synthesis problem, and classical lower bounds:

\begin{theorem}[\Cref{cor:efi_bpp_nexp}, informal]
    \label{thm:main-informal}
    Suppose quantum commitments exist.
    Then either:
    \begin{enumerate}
        \item There is no polynomial-time solution to the unitary synthesis problem, or
        \item $\BPP \neq \NEXP$.
    \end{enumerate}
\end{theorem}

Recall that $\BPP$ is the set of problems solvable in randomized polynomial time, while $\NEXP$ is nondeterministic \textit{exponential} time.
Some might be surprised to find that $\BPP \stackrel{?}{=} \NEXP$ is a question that we do not yet know the answer to, but it has been open ever since the classes were first defined and remains one of the frontier questions in classical complexity theory.
Consequently, any unconditional proof that quantum commitments exist would have to resolve at least one of a 20+ year old open problem~\cite{AK07-qcma-qma} or a 40+ year old one~\cite{Hel86-bpp-exp,BFL91-mip-nexp}!

We prove \Cref{thm:main-informal} as a consequence of the following key lemma, stated in terms of unitary complexity classes:

\begin{restatable}[\Cref{lem:unitarypspace_dichotomy_witnexp}, informal]{lemma}{lemintrokey}
    \label{lem:all-fnexp-informal}
    For every $\calU \in \uPSPACE$, either
    \begin{enumerate}
        \item $\calU \not\in \uP^\ALL$, or
        \item $\calU \in \uP^\FNEXP$.
    \end{enumerate}
\end{restatable}

Here, $\uP^\class{K}$ is the set of unitary problems that can be solved in polynomial time with an oracle from $\class{K}$, and $\ALL$ is the set of all languages.
Thus, $\uP^\ALL$ is shorthand for the set of unitaries that \textit{do} admit an efficient unitary synthesis algorithm relative to some oracle.
(Indeed, in the above notation, one can express the unitary synthesis problem succinctly as $\uP^\ALL \stackrel{?}{=} \uALL$, where $\uALL$ is the set of all unitaries.)
For the second item, $\FNEXP$ denotes the set of $\NEXP$ search problems, and an oracle for $\FNEXP$ means the ability to query the bits of some valid $\NEXP$ witness (see \Cref{sec:nexp_relations} for more details).
Hence, \Cref{lem:all-fnexp-informal} shows that for every unitary $\calU$ implementable in polynomial space, either $\calU$ is ``unitary synthesis-hard'' in the sense that it cannot be implemented efficiently relative to any classical oracle, or $\calU$ reduces efficiently to $\NEXP$ witness search.

Lombardi, Ma, and Wright~\cite{LMW24-synthesis} once remarked that we do not even know whether an oracle for the \textit{halting problem} suffices to break every quantum commitment scheme, or more generally to synthesize unitaries in $\uPSPACE$.
While true, \Cref{lem:all-fnexp-informal} shows that that statement is perhaps misleading, as for every unitary in $\uPSPACE$, either a halting oracle is as useless as any other, or it is substantial overkill.

Another way to interpret \Cref{lem:all-fnexp-informal} is as a statement that a positive resolution to unitary synthesis would immediately place an upper bound on the power of an oracle needed to implement $\uPSPACE$ efficiently:

\begin{corollary}[\Cref{cor:u-all-up-all}, informal]
    \label{cor:upspace_fnexp}
    If $\uALL = \uP^\ALL$, then $\uPSPACE \subseteq \uP^\FNEXP$.
\end{corollary}

\cite{BEMPQY26-uhlmann} raised the question of how unitary complexity classes formally relate to their classical counterparts.
As one concrete example, they asked: does $\P = \PSPACE$ imply $\uP = \uPSPACE$, or can one give oracle evidence that it does not?
More generally, one could also ask about a ``$\PSPACE$-effective'' version of unitary synthesis: $\uPSPACE \stackrel{?}{=} \uP^\PSPACE$.
Indeed, if polynomial-space unitaries \textit{can} be synthesized efficiently, it is natural to want the classical oracle to also be computable in polynomial space.
Assuming the unitary synthesis conjecture, \cref{cor:upspace_fnexp} almost obtains ``$\PSPACE$-effective'' synthesis, but with the $\PSPACE$ oracle replaced by $\FNEXP$.
Nevertheless, we emphasize that $\FNEXP$ is not known to be separated from $\PSPACE$.
For example, if $\AM = \NEXP$, then the $\FNEXP$ oracle can be replaced by $\PSPACE$ as a consequence of the \textit{easy witness lemma} of Impagliazzo, Kabanets, and Wigderson~\cite{IKW02-easy-witness}.
Thus \cref{cor:upspace_fnexp} gives a condition under which $\PSPACE$-effective unitary synthesis reduces to full unitary synthesis.
Or, stated another way, we have:

\begin{corollary}[\Cref{cor:nc_np_separation}, informal]
    \label{cor:am_nexp_informal}
    If $\uP^{\PSPACE} \neq \uPSPACE$, then at least one of the following holds:
    \begin{enumerate}
        \item $\uP^{\ALL} \neq \uALL$, or
        \item $\AM \neq \NEXP$.
    \end{enumerate}
\end{corollary}

We remark that most of our results relativize in some fashion, and in the formal statements of our results, we will always be completely explicit about how the statements hold relative to oracles.
In one case, we even show that \textit{a mere oracle separation} involving $\uPSPACE$ would imply new non-oracular lower bounds:

\begin{corollary}[\Cref{cor:nc_np_separation}, informal]
    \label{cor:nc_np_informal}
    If there is a language $L$ relative to which $\uP^{\PSPACE^L} \neq \uPSPACE^{L}$, then either
    \begin{enumerate}
        \item $\uP^{\ALL} \neq \uALL$, or
        \item $\NC \neq \NP$
    \end{enumerate}
\end{corollary}

Here, $\NC$ is the set of problems solvable in polylogarithmic time on a machine with a polynomial number of parallel processors.
Like most complexity separations, whether $\NC \stackrel{?}{=} \NP$ is wide open, and we emphasize that the second implication does \textit{not} involve relativization.
This is noteworthy because \cite{BEMPQY26-uhlmann} left open the problem of constructing a classical oracle $L$ relative to which $\P^L = \PSPACE^L$ but $\uP^L \neq \uPSPACE^L$.
\Cref{cor:nc_np_separation} explains why such an oracle has remained elusive: it would either have to resolve the unitary synthesis problem, or yield a breakthrough lower bound in classical complexity!

\section{Detour: (re)defining unitary complexity classes}
\label{sec:detour}

It has proven challenging to settle upon good definitions of unitary complexity classes, in part because they behave less robustly than traditional complexity classes such as $\P$ or $\BQP$.
Indeed, there are many reasonable-looking choices one can make in the definitions that lead to starkly different notions.
As one example, unitary complexity classes are typically defined allowing some kind of error in the implementation, not just because we want to capture probabilistic algorithms, but also because often the gates in a quantum circuit are discrete even though the unitary families we want to implement are continuous.
Yet unlike with (say) $\BQP$, there is no generic method for error \textit{reduction}: the standard technique of amplifying a Monte Carlo algorithm by repetition plus majority voting has no analogue in the unitary setting.
So, unitary complexity classes must make a choice for how to bake error reduction into the definition.
The many possible choices have led to numerous different definitions for the complexity classes $\uP$ (often denoted $\class{unitaryBQP}$\footnote{We find it strange and inconsistent that the ``$\class{BQ}$'' prefix typically makes its way into $\class{unitaryBQP}$ but not $\uPSPACE$. As the ``$\class{Q}$'' meaning ``quantum'' is made redundant by ``$\class{unitary}$,'' we prefer the more succinct notation.}) and $\uPSPACE$~\cite{RY22,MY23,BEMPQY26-uhlmann,Zha25-unitary-oracles}.

\subsection{Why do we need new definitions?}

Despite the variety of definitions for unitary complexity classes, all existing proposals present serious issues.
Disconcertingly, it is not clear whether the results claimed in these works are insensitive to the selections made in their respective definitions.
If we wish to put unitary complexity on solid foundations, we have to ensure that the definitions that we choose are sensible.
We now describe the most critical issues with prior definitions: these involve showing that one can ``smuggle'' a hard-to-compute unitary into a class which intuitively should not capture such a problem.

Rosenthal and Yuen~\cite{RY22} and Metger and Yuen~\cite{MY23} were some of the earliest works to study unitary complexity, but the first formal treatment of unitary complexity classes was defined and investigated in the work of Bostanci, Efron, Metger, Poremba, Qian, and Yuen~\cite{BEMPQY26-uhlmann}.
The central trouble with their definition is with how it treats error: an $n$-qubit unitary is in $\uP$ as defined by them, when there is a circuit of size $\poly(n, 1/\eps)$ which approximates the unitary to $\eps$ error.\footnote{The arXiv v1 version of \cite{BEMPQY26-uhlmann} used a different definition with further issues that were discovered by the first author, Fermi Ma, and John Wright, and subsequently raised to the authors.
The arXiv v3 of their preprint details these objections, responds, and makes changes to address some of them.
The class as discussed here is based on their most recent definition.}
This definition is self-consistent until we introduce a notion of $\uPSPACE$, which is defined analogously, except that the circuit must merely use space $\poly(n, 1/\eps)$.
At this point, we can try composing different unitary complexity classes: for example, we expect that $\uPSPACE$ is not strengthened when we provide it with a $\uP$ oracle; in other words, we expect that $\uPSPACE^\uP = \uPSPACE$, as is true for the non-unitary versions of these classes.
More generally, we should even expect that $\uPSPACE$ is self-low, meaning $\uPSPACE^\uPSPACE = \uPSPACE$.
However, this is not the case under the conventional notion of oracle access: consider the task of, given an $x \in \{0,1\}^n$, applying the unitary
\begin{equation}
    \label{eq:u_x_exp}
    U_x = \begin{bmatrix}
        1 \\ & e^{\ii \pi L(x) / 2^n}
    \end{bmatrix},
\end{equation}
where $L(x)$ is a function which requires exponential time and space to compute.
Then this unitary is in their $\uP$, as it is close to the identity for large $\eps$, and the phase can be computed in $\poly(1/\eps)$ time for $\eps$ exponentially small in $n$.
However, if we give a polynomial-space circuit access to $U_x$, it can compute $U_x^{2^n} = Z^{L(x)}$ and therefore implement a unitary which requires exponential space.

\cite{BEMPQY26-uhlmann} handles this corner case by switching to an extremely non-standard notion of oracle access: they specify that, when a class is given an oracle, the resulting circuit must satisfy the error guarantee, even when the oracle is replaced with \emph{an adversarially chosen perturbation of the oracle}.
Encoding robustness into the notion of the oracle is an unnecessary concession: if instead we required the time/space complexities to scale polynomially with $\log(1/\eps)$ instead of $1/\eps$, this issue does not arise, and the statement $\uPSPACE^\uP = \uPSPACE$ holds for the conventional notion of oracle.
However, the authors settle for these infelicities, as their main results critically rely on the more generous error scaling of their definitions.

The other prominent suite of definitions was proposed by Zhandry~\cite{Zha25-unitary-oracles}.
Zhandry points out that having complexity scale inverse polynomially with error is unsuitable for cryptographic applications, and proposes three alternative definitions of $\uP$/$\uPSPACE$, where a circuit for implementing $U$ is required to have either arbitrarily-small inverse-polynomial, negligible, or arbitrarily-small inverse-exponential error in $n$.
Curiously, the dependence on the error scaling is handled in a non-uniform manner, meaning that the length of source code describing a quantum algorithm could blow up arbitrarily between (say) $\eps = 1/n$ to $\eps = 1/{2n}$.
Though, inverse-exponential error does address some of the issues of the prior work.
However, having error scale with $n$ (as opposed to being a parameter which one can choose independently) again leads to issues under composition.
As before, we can hide a difficult problem in a slight phase, such that it can be implemented to inverse-exponential error, but when this is given as an oracle to a more powerful class which can apply it repeatedly, the difficult problem is revealed.
Consider the example of, given $x \in \{0,1\}^n$, applying
\begin{equation*}
    U_x = \begin{bmatrix}
        1 \\ & e^{\ii \pi L(x) / 2^{2^n}}
    \end{bmatrix},
\end{equation*}
but now let $L(x)$ be an undecidable function, such as the halting problem.
Then the identity implements $U_x$ to any inverse-exponential error $\eps = 2^{-\poly(n)}$.
However, a polynomial-time algorithm given the padded input $x | 1^{2^{2^{|x|}}}$ can use an oracle for $U_x$ to compute $U_x^{2^{2^n}} = Z^{L(x)}$, and thus solves an undecidable problem.
The failure of ``closure under padding'' is resolved by requiring the existence of an efficient circuit for all $\eps$, not just for $\eps$ with a particular scaling in $n$.

\subsection{Motivating our definitions}

We wanted to ensure that our results rest on rigorous footing, rather than leaving any possibility that they are the consequences of arbitrary choices.
As such, we consider it essential to derive a set of unitary complexity definitions forward from first principles instead of backward from results.
We posit that a good complexity class should satisfy several properties:

\begin{enumerate}
    \item \textbf{Robust:} the class should remain the same, even if the definition is altered slightly, or the underlying implementation details are changed.
    \item \textbf{Consistent:} definitions should follow a common pattern across similar classes.
    \item \textbf{Intuitive:} the class should behave consistently with our intuition about the types of computations we are trying to capture.
    For example, it should include the kinds of problems which we expect it to include.
    \item \textbf{Composable:} fundamental ``composition'' operations like oracle reductions should be natural and produce sensible outcomes.
\end{enumerate}

These desiderata will guide us as we make our definitions.
We will discuss the reasoning behind how we make our choices, and point out when our choices (1) are cosmetic, so a different choice would produce the same class; (2) lead to materially different but reasonable classes; or (3) are critical to making well-behaved classes.

In this work, we consider $\uP$ and $\uPSPACE$, which ought to have the following colloquial definition.

\begin{center}
    $\uP$/$\uPSPACE$ is the class of \textit{unitary languages} which \\ can be \textit{implemented} in \textit{polynomial time/space} on a \textit{quantum computer}.
\end{center}

We now formalize this definition, piece by piece.

\paragraph{``unitary language'':}
For us, a unitary problem is a sequence of unitaries of increasing size, $\{U_n\}_{n \in \N}$ where $U_n$ is a unitary on $n$ qubits (\cref{def:u-language}).
This sequence of unitaries can be thought of as defining a problem for every input size, where the problem is to apply the appropriate unitary to the input state.
For example, one might define the unitary language \textsc{QFT} of applying the $n$-qubit quantum Fourier transform.

Previous works~\cite{BEMPQY26-uhlmann,Zha25-unitary-oracles} define unitary problems differently: they index their unitaries by strings $x \in \{0,1\}^*$.
This does not affect the expressivity of unitary languages.
Languages in our sense can be converted to languages in their sense by taking $U_x = U_{\abs{x}}$, and languages in their sense can be converted into languages in our sense by defining $U_{2n}$ to be the unitary which takes $\ket{x}\ket{\psi}$ to $\ket{x} U_x \ket{\psi}$.
However, it is often useful to apply the unitary to a coherent superposition of inputs $x$, which makes the latter notion more convenient.

\cite{BEMPQY26-uhlmann} also allow their implementation problems to consist of \textit{partial isometries}, which only need to implement a given unitary on some promised subspace.
Hence, partial isometries generalize unitary transformations in the same way that promise problems generalize languages.
While it would be interesting to generalize our definitions in the same fashion, we choose to focus on the conceptually simplest case.

Because we choose to only index by input size, we do not have a hard-coded classical input to our problems.
Note that we cannot make our choice simpler by discarding the index over $n$, since we need to apply a different map for different input sizes.

\paragraph{``implemented'':}
Here, we make two different choices about how to formalize this.
For us, to implement a unitary $U_n$ is to be able to, for every $\eps$, find a unitary quantum circuit $C_{n, \eps}$ such that, for all $n$-qubit states $\ket{\psi}$, one of the following two conditions holds:
\begin{align*}
    C_{n,\eps}(\ket{\psi}\ket{0^a}) &\approx_\eps (U\ket{\psi}) \ket{0^a} \\
    C_{n,\eps}(\ket{\psi}\ket{0^a}) &\approx_\eps (U\ket{\psi}) \ket{g} \text{ for some } \ket{g} \text{ independent of } \ket{\psi}
\end{align*}
The first definition requires that our implementation return the ancilla qubits to their initial state; the second is not concerned with this choice, and asks only that $U$ be correctly applied to the first register.
We refer to the resulting classes as $\uP/\uPSPACE$ and $\puP$/$\puPSPACE$, respectively.
Note that the first definition maintains global phase information on $U$, but the second definition allows the circuit to implement $U$ with any global phase, because the phase can be viewed as part of the garbage state $\ket{g}$.
The forgetting of global phase is the sense in which $\puP$ is ``projective,'' referring to the projective unitary group of unitary matrices modulo phase.
However, it is not \textit{a priori} clear whether the two definitions produce the same class after disregarding global phase.

The choice between allowing or disallowing garbage was the most contentious topic among the authors of this paper.
Both appear to be valid choices, both could be argued to capture the notion of ``implementing'' a unitary, and both have precedent in prior work, with \cite{Zha25-unitary-oracles} using the clean notion and \cite{BEMPQY26-uhlmann} allowing garbage.
The garbage-free definition has the advantage that it is, well, \textit{cleaner}.
And even though global phase is unobservable, it is often convenient to keep track of it in quantum algorithms, e.g.\ in linear combinations of unitaries~\cite{cw12}.
Certain operations, such as mapping $U$ to its controlled operation $\controlled U$, only make sense when the global phase is defined.
Zhandry uses precisely these arguments to favor $\uP$ over $\puP$, but we chose to include both, for three reasons.
First, we find it hard to deny that implementing a unitary with garbage should count as implementing that unitary.
Second, we cannot rule out that the set of unitaries which can be implemented efficiently with garbage is substantively larger than those which can be implemented cleanly.
(In fact, we show that relative to unitary oracles, allowing implementations with garbage yields \textit{strictly more} power; see \Cref{thm:no_global_phase_informal} further below.)
Finally, provided that we generally only work with $\puP$ modulo the global phase, there are no complexity-theoretic inconsistencies.

A naive choice would be to not allow any error, and instead require an exact implementation of $U_n$ (as if we took $\eps = 0$ above).
However, this would mean that changing the underlying gate set used in the quantum computer would yield different definitions of $\uP$.
So, we allow some error in implementation.

On the other hand, the choice of error metric turns out not to matter.
Up to polynomial overhead, the definitions above are equivalent to other sensible definitions.
In particular, the clean definition is equivalent to requiring that $C_{n, \eps}$ is close to an \emph{extension} of $U$ (i.e., some $\overline{U}$ such that $\overline{U} \ket{\psi} \ket{0^a} = (U \ket{\psi}) \ket{0^a}$) in operator norm.
The garbage definition is equivalent to requiring that $C_{n, \eps}$ is close to a \emph{dilation} of $U$ (i.e., some $\overline{U}$ such that $\overline{U} \ket{\psi} \ket{0^a} = (U \ket{\psi}) \ket{g}$) in diamond norm or operator norm.
It is also equivalent to requiring that $C_{n, \eps}$ is close to $U$ in diamond norm upon tracing out the ancilla register.
Note that these are all notions of \emph{worst-case} error, where we would like to be close for all inputs.

Technically, our earlier choice to index by length $n$ rather than strings $x \in \{0,1\}^n$ could materially affect the definitions of $\puP$ and $\puPSPACE$, as the ability to implement the map $\ket{\psi} \to U_x\ket{\psi}$ individually for every $x$ might not let us coherently apply the map $\ket{x}\ket{\psi} \to \ket{x}U_x\ket{\psi}$ on a superposition over $x$; the latter requires that the same garbage state is produced for every $\ket{x}$.
Indeed, \cite[Section 3.6.1]{BEMPQY26-uhlmann} raises the same concern and uses it to argue for indexing by classical input strings.
However, it will turn out that even if different $x$ yield different garbage states, there is another way that we can coherently encode all of the unitaries $U_x$ into a single larger unitary.
Namely, one can coherently apply the map $\ket{x}\ket{\psi}\ket{\phi} \to \ket{x}U_x\ket{\psi}U_x^\dagger\ket{\phi}$ (a consequence of \Cref{thm:garbage_iff_informal} further below).

\paragraph{``polynomial time/space'':}

We require that our algorithms run in time/space polynomial in $n$, the input size, as well as $\log 1/\eps$.
This differs from the choice of~\cite{BEMPQY26-uhlmann}, which only requires polynomial time/space in $n$ and $1/\eps$.
Zhandry is agnostic to this choice, making three different classes according to the scaling of $\eps$ as a fixed function of $n$.

Another difference here compared to prior work is to require \textit{uniform} dependence on both $n$ and $\eps$: we assume the existence of a polynomial-time/space Turing machine $M$ such that $M(1^n, \eps)$ outputs a description of a quantum circuit $C_{n, \eps}$ with the desired guarantee.
The Turing machine has a write-only output tape that specifies $C_{n, \eps}$ one gate at a time; this allows a polynomial-space machine to specify an exponentially-long circuit, so long as the circuit uses only a polynomial number of qubits.

The choices of error dependence are the subject of the greatest disagreement among different definitions: \cite{BEMPQY26-uhlmann} uses polynomial scaling with $\eps$ as a uniform parameter, while \cite{Zha25-unitary-oracles} argues for negligible or inverse-exponential error with circuits only specified for $\eps$ scaling with a fixed function of $n$, non-uniformly.
Our choice of scaling is the only option that leads to well-behaved definitions.
As mentioned before, choosing a larger scaling than $\polylog(1/\eps)$ runs into issues with composition.
In particular, self-lowness of $\uPSPACE$ fails if the allowed runtime is $s(\log 1/\eps)$ for some superpolynomial $s$, by considering the same $U_x$ from \Cref{eq:u_x_exp} but with $L(x)$ a hard language for space $s(n)$.
On the other hand, if the error scaling were \textit{smaller} than $\log(1/\eps)$, then $\uP$ would no longer be robust to changing the gate set, as the Solovay--Kitaev theorem requires $\Omega(\log 1/\eps)$ gates to compile an arbitrary quantum gate into a universal gate set.
Additionally, if $\eps$ scales with $n$ and cannot be scaled arbitrarily, then self-lowness fails again, and the complexity class is no longer closed under padding: problems become \textit{harder} when they are padded, an unintuitive outcome.

\cite[Section 3.6.3]{BEMPQY26-uhlmann} responds to these objections by contending that the inverse polynomial scaling allows their classes to capture algorithms like tomography and FPTASes.
However, if the goal is to build analogues to $\P$ and $\PSPACE$, the most fundamental and basic complexity classes, these choices are odd (for example, the knapsack problem is $\NP$-hard but also captured in $\uP$?), especially when they produce ill-behaved classes.
Though $\polylog(1/\eps)$ scaling is a strong guarantee, the most important unitary synthesis problems all meet it, including classical error reduction for decision problems, Hamiltonian simulation, and QSVT~\cite{gslw18}.

\paragraph{``quantum computer'':}

We model quantum algorithms as quantum circuits consisting of unitary gates.
In the case of $\uP$ and $\uPSPACE$, this is necessary for the definition to make sense, because non-unitary channels do not maintain a meaningful notion of global phase.
For $\puP$, this is without loss of generality, because non-unitary gates can be replaced by unitary dilations at the cost of at most a polynomial increase in the number of qubits.
However, the situation for $\puPSPACE$ is less obvious, because purifying the computation would blow up the space complexity by an exponential.
As it turns out, there is a nontrivial method to purify a non-unitary implementation of a unitary transformation while maintaining polynomial space (see \Cref{rmk:pspace_purification} for details).
Thus it is without loss of generality to make all of the gates unitary, which we do for aesthetic simplicity.

We further assume the gates have explicitly-given entries with a \emph{fully polynomial time exponential approximation scheme} (FPTEAS), roughly meaning that they can be any numbers whose digits are computable in polynomial time.
More specifically, gates should be approximable to error $\eps$ in time $\polylog(1 / \eps)$.
Examples of gates with an FPTEAS include all gates with algebraic entries.

Our choice borrows from Kuperberg~\cite{Kup15-jones}, who suggested using FPTEAS-computable gates after observing that the class $\class{PostBQP}$ is not robust when allowed gates that require larger time to evaluate to error $\eps$.
$\uPSPACE$ suffers from a similar issue as $\class{PostBQP}$; one way to see this is to observe that there exists a space-$\polylog(1/\eps)$ algorithm that queries an $O(1)$-qubit unitary $U$ and approximates the entries of $U$ to error $\eps$, by performing process tomography~\cite{HKOT23-diamond}.
So, if we allowed gates $U$ whose entries did \textit{not} have an FPTEAS, e.g.\ those requiring more than $\polylog(1/\eps)$ time and space, $\uPSPACE$ would again smuggle in the overly-hard problem of computing the entries of $U$.

On the other hand, our choice to allow \textit{any} FPTEAS is convenient because it allows us to abstract away the details of compilation into a fixed gate set, while still allowing access to gates with irrational algebraic entries such as the Hadamard gate.
For comparison, Zhandry's assumption of an arbitrary universal gate set leads to the annoyance of having to keep track of circuit compilation in every proof~\cite{Zha25-unitary-oracles}.
Under that definition, operations as simple as replacing a circuit by its complex conjugate require recompiling the entire circuit and accounting for the error incurred.
Nevertheless, we emphasize that allowing \textit{arbitrary} gates with an FPTEAS does not make our definition any more general than it would be with a fixed gate set, such as Clifford + T.
This is because the Solovay--Kitaev theorem can produce an error-$\eps$ compilation of any FPTEAS-computable gate in time $\polylog(1 / \eps)$~\cite[Theorem 2.4]{Kup15-jones}.

\paragraph{Oracles and reductions:}

Upon defining $\uP$ and $\uPSPACE$, the last step is to augment the classes to allow oracle access to a unitary language, in order to have a meaningful notion of composability.
A first try would be to simply allow circuits to apply elements $U$ of the language provided to it; this is how~\cite{BEMPQY26-uhlmann} handle oracles.
However, this is a fairly weak type of access: many kinds of subroutines, including QSVT~\cite{gslw18}, amplitude amplification~\cite{BHMT02-amp}, and phase estimation~\cite{kitaev95}, also demand access to $U^\dagger$, $\controlled U$, or $\controlled U^\dagger$.
It is reasonable to grant these demands, as any circuit to instantiate the oracle also immediately provides circuits for these alternative oracles.

Zhandry~\cite{Zha25-unitary-oracles} goes even further, arguing that one should provide access to the more obscure $U^T$ and $U^*$, since the same logic goes for them, and they are, in a formal sense, the \textit{complete} list of oracles which need to be granted following this logic.
In our definitions, we follow Zhandry's generality, with the exception that for $\puP$ and $\puPSPACE$, we do not (and indeed, cannot) give access to the controlled unitaries, as this would lift an undistinguished global phase to a distinguishable relative phase.

Deciding which oracles to choose for the definition was another topic of contention among the authors.
Ultimately, we decided it was best to opt for the most permissive definition of oracles.
As a counterpoint, there are settings where the oracle is not instantiated by a circuit, but instead by nature, or another party keeping their circuit secret, or some other manner where one does not have access to all possible modified oracles.
However, like many contentions, this one is fairly inconsequential: our results hold regardless of what one takes oracle access to mean, with the exception that both $U$ and $U^\dagger$ seem to be required.

Our notion of oracle leads to one more issue which must be addressed: it might appear that $\puPSPACE$ is not self-low.
After all, if an oracle is called exponentially many times, and we instantiate the oracle with a circuit that leaves garbage after every call, then the garbage would appear to pile up and eventually exceed the polynomial space of a unitary circuit!
We avoid this issue by showing that given forward and inverse access to a dilation of an $n$-qubit unitary $U$, one can simulate arbitrarily many calls to $U$ while only holding at most $n$ qubits of garbage at any stage (another consequence of \Cref{thm:garbage_iff_informal} further below).

\paragraph{Summary:} A comparison of the various definitions is summarized in \Cref{tab:unitary-definitions}.

\begin{table*}[t]
\centering
\small
\renewcommand{\arraystretch}{1.35}
\setlength{\tabcolsep}{5pt}

\rowcolors{3}{gray!7}{white}

\begin{tabularx}{\textwidth}{
    >{\raggedright\arraybackslash}p{0.16\textwidth}
    >{\raggedright\arraybackslash}X
    >{\raggedright\arraybackslash}X
    >{\raggedright\arraybackslash}X
    >{\raggedright\arraybackslash}X
    >{\centering\arraybackslash}p{0.09\textwidth}
}
\toprule
\rowcolor{white}
\textbf{Feature}
    & \multicolumn{2}{c}{\textbf{Previous definitions}}
    & \multicolumn{2}{c}{\textbf{Our definitions}}
    & \textbf{Relevance} \\
\cmidrule(lr){2-3}
\cmidrule(lr){4-5}

\rowcolor{white}
    & \textbf{\cite{BEMPQY26-uhlmann}}
    & \textbf{\cite{Zha25-unitary-oracles}}
    & \textbf{$\uP$}
    & \textbf{$\puP$}
    & \\
\midrule

Indexed by
    & Classical input
    & Classical input
    & Length only
    & Length only
    & Cosmetic \\

Allowed unitary gates
    & H, CNOT, T
    & Any universal gate set
    & Any with FPTEAS entries
    & Any with FPTEAS entries
    & Cosmetic \\

Error metric $\varepsilon$
    & Diamond distance
    & Operator distance (to extension)
    & Operator distance (to extension)
    & Operator distance (to dilation)
    & Cosmetic \\

$\varepsilon$ independent of $n$
    & Originally no; later yes
    & No
    & Yes
    & Yes
    & Essential \\

Scaling in $\varepsilon$
    & $\poly(1/\varepsilon)$
    & Various
    & $\polylog(1/\varepsilon)$
    & $\polylog(1/\varepsilon)$
    & Essential \\

Garbage
    & Yes
    & No
    & No
    & Yes
    & Material \\

Allowed oracle queries
    & $U$ only
    & $U$, $U^\dagger$, $U^*$, $U^\trans$ +~controlled
    & $U$, $U^\dagger$, $U^*$, $U^\trans$ +~controlled
    & $U$, $U^\dagger$, $U^*$, $U^\trans$
    & Material \\

Noisy oracle queries
    & Yes
    & No
    & No
    & No
    & Essential \\

\bottomrule
\end{tabularx}

\caption{Comparison of definitions of efficient unitary computation.
The features can be divided into choices that are essential to reach a reasonable definition, materially consequential but arbitrary, and purely cosmetic.}
\label{tab:unitary-definitions}
\end{table*}

\section{Technical overview and further results}

Here we outline the progression of the paper in rough chronological order, making note of several other noteworthy results that we obtain along the way to our main theorems.

\subsection{Comparing unitary and projective unitary classes}

It might seem unsatisfactory or even elusive that we supply two sets of definitions rather than settling upon a single ``correct'' one.
While a valid concern, it turns out that the two notions of efficient unitary computation (with and without garbage) are more closely related than they might seem.
In fact, both definitions can be viewed as special cases of each other!
Concretely, we show:

\begin{theorem}[\Cref{thm:u_garbage_iff_u_udag_tidy}, informal]
    \label{thm:garbage_iff_informal}
    $\calU \in \puP$ if and only if $\calU \otimes \calU^\dagger \in \uP$.
\end{theorem}

And conversely:

\begin{theorem}[\Cref{thm:u_tidy_iff_cu_garbage}, informal]
    \label{thm:control_iff_informal}
    $\calU \in \uP$ if and only if $\controlled\calU \in \puP$.
\end{theorem}

Taken together, we see that one could use either $\uP$ or $\puP$ when proving some result, because any statement involving one class has an equivalent reformulation in terms of the other.
Thus, one should feel free to choose whichever definition is more convenient.
Typically, this will be $\uP$, because juggling garbage states is more cumbersome.

In the proof of \Cref{thm:garbage_iff_informal}, the more interesting direction is the forward one.
The basic idea is that if $C$ implements a dilation of $U$ with garbage state $\ket{g}$, then running $C$ \textit{backwards} with $\ket{g}$ as the ancilla both applies $U^\dagger$ and erases the garbage.
For \Cref{thm:control_iff_informal}, again the nontrivial direction is from $\puP$ to $\uP$.
Since $\controlled U$ always has the all-zeros state as an eigenstate, one can again run $C$ backwards on an all-zero input to erase the garbage.
Both theorems additionally hold with $\P$ replaced by $\PSPACE$, or relative to quantum oracles (or both); see \Cref{sec:garbage_tidy_relations} for details.

Interestingly, one can also show:

\begin{theorem}[\Cref{thm:u_garbage_iff_u_ustar_tidy}, informal]
    \label{thm:garbage_star_informal}
    $\calU \in \puP$ if and only if $\calU \otimes \calU^* \in \uP$.
\end{theorem}

So \Cref{thm:garbage_star_informal} looks the same as \Cref{thm:garbage_iff_informal}, but with the conjugate transpose replaced by just the complex conjugate!
In hindsight, this possibility should not surprise us, as both $U \otimes U^\dagger$ and $U \otimes U^*$ are examples of canonical unitary transformations that ``forget'' the global phase information in $U$.
(By contrast, note that $U \otimes U^\trans$ is \textit{not} phase-independent.)
The proof idea once again is to erase garbage using a known eigenstate.
Specifically, we use the fact that the maximally-entangled state $\frac{1}{\sqrt{2^n}} \sum_{x \in \{0,1\}^n} \ket{x}\ket{x}$ is a $1$-eigenstate of $U \otimes U^*$ for any $n$-qubit $U$.

\subsection{Characterizing space-bounded unitary computation}

For polynomial-\textit{space} computation, we establish an even tighter connection between the unitary and projective-unitary definitions.
In particular, we prove that one can generically erase garbage:

\begin{theorem}[\Cref{thm:projunitarypspace_equiv}, informal]
    \label{thm:projunitarypspace_equiv_informal}
    $\calU \in \puPSPACE$ if and only if there exists a $\calV \in \uPSPACE$ that equals $\calU$ up to global phase.
\end{theorem}

\Cref{thm:projunitarypspace_equiv_informal} is the strongest possible equivalence, because a unitary language $\calU \in \puPSPACE$ may encode uncomputable information in the phases, and therefore we cannot directly conclude $\calU \in \uPSPACE$.
Still, \Cref{thm:projunitarypspace_equiv_informal} is general enough to justify using only $\uPSPACE$ when proving statements about space-efficient unitary computation, as it says we can always assign \textit{some} global phase that makes the unitary computable cleanly.

Perhaps more useful than \Cref{thm:projunitarypspace_equiv_informal} itself is the proof by which we obtain it.
The key insight is that $\uPSPACE$ has an equivalent characterization in terms of the computability of its entries:

\begin{theorem}[\Cref{thm:entrywise_characterization}, informal]
    \label{thm:entrywise_characterization_informal}
    $\calU \in \uPSPACE$ if and only if each entry of the $n$-qubit unitary in $\calU$ can be approximated to error $\eps$ in space $\poly(n, \log 1/\eps)$.
\end{theorem}

\Cref{thm:entrywise_characterization_informal}
yields \Cref{thm:projunitarypspace_equiv_informal} by combining with the equivalence
\[
\calU \in \puPSPACE \iff \calU \otimes \calU^\dagger \in \uPSPACE, \qquad (\mathrm{\Cref{thm:garbage_iff_informal}})
\]
observing that we can use the entries of $\calU \otimes \calU^\dagger$ to compute the entries of $\calU$ up to a canonical global phase of our choice.

For the proof of \Cref{thm:entrywise_characterization_informal}, the forward direction is essentially a consequence of the equivalence $\BQPSPACE = \PSPACE$~\cite{Wat99-space,Wat03-space,FR21-space}.
That is, for \textit{decision} problems, quantum and classical polynomial space are equivalent.
So, it suffices to find a polynomial-space quantum algorithm to estimate the entries of $\calU$, which amplitude estimation accomplishes~\cite{BHMT02-amp}.
For the reverse direction, we first use \textit{state synthesis}~\cite{GR02-state-prep,Aar16-barbados,INNRY22-qma-search} to show that we can space-efficiently prepare any column of $\calU$, given access to the entries.
Then, we appeal to a proof by Rosenthal~\cite{Ros21-unitary} that turns a $\calU$-column constructor into an implementation of $\calU$, with polynomial space overhead.

Of note, when working toward the proof of \Cref{thm:entrywise_characterization_informal}, we find it helpful to work with certain unitary languages that encode real numbers in the entries.
In particular, we will frequently use the following definition of a \textit{phase encoding}:

\begin{definition}[Restatement of \Cref{def:fphase}]
    Let $f : \{0,1\}^* \to [-\pi/2, \pi/2]$.
    The \emph{phase encoding of $f$} is the unitary language $f\Phase = \{R_{f,n}\}_{n \in \N}$ defined by
    \[
    R_{f,n}\ket{x}
    =
    e^{i f(x)}\ket{x},
    \]
    for every $x \in \{0,1\}^n$.
\end{definition}

Phase estimation~\cite{kitaev95,NC10-book} almost immediately implies that the ability to implement $f\Phase$ is tightly connected to the approximability of $f$:

\begin{lemma}[\Cref{lem:rotation_to_fqpseas}, informal]
    $f\Phase \in \uPSPACE$ if and only if $f(x)$ can be approximated to error $\eps$ in space $\poly(|x|, \log 1/\eps)$.
\end{lemma}

When $f\Phase \in \uPSPACE$, it is particularly convenient to treat $f\Phase$ as an oracle within another $\uPSPACE$ algorithm, because we do not need to account for approximation errors when querying it; the self-lowness proof of $\uPSPACE$ handles this directly!
We defer to \Cref{sec:characterization} for the full details, but emphasize that phase encodings are an abstraction whose utility and convenience only became apparent to us in light of the composability of $\uPSPACE$.
See \Cref{lem:entry_to_column} for a particularly illustrative example of the elegant proofs that our definitions afford.

We conclude our discussion of $\uPSPACE$ by remarking that \Cref{thm:projunitarypspace_equiv_informal} also holds relative to classical oracles, but provably \textit{fails} relative to quantum oracles!

\begin{theorem}[\Cref{thm:no_global_phase}, informal]
    \label{thm:no_global_phase_informal}
    There is a unitary language $\calU$ such that for every unitary language $\calV$ that equals $\calU$ up to global phase, we have $\calV \not\in \uPSPACE^{\calU \otimes \calU^\dagger}$.
\end{theorem}

On the other hand, we clearly have $\calU \in \puP^{\calU \otimes \calU^\dagger}$, because $\calU^\dagger\ket{0^n}$ can be taken as the garbage state $\ket{g}$.
The proof of \Cref{thm:no_global_phase_informal} hinges on an observation made by Gavorov\'{a}, Seidel, and Touati~\cite{gst24} that there is no continuous map from the projective unitary group $PU(N)$ to the unitary group $U(N)$ that assigns a global phase to every projective unitary.
Moreover, \Cref{thm:no_global_phase_informal} relies on very little about $\uPSPACE$ other than the fact that it must approximately implement $\calV$ with a finite number of queries to $\calU \otimes \calU^\dagger$ (see \Cref{rmk:topology}).
Thus, \Cref{thm:no_global_phase_informal} shows that relative to unitary oracles, erasing garbage is generally impossible, and $\puP$ can be strictly more powerful than \textit{any} ``clean'' unitary class!

\subsection{Relation to unitary synthesis}

At last, we may return to the results described in \Cref{sec:this_work}.
Recall our key lemma:

\lemintrokey*

The entrywise characterization of $\uPSPACE$ (\Cref{thm:entrywise_characterization_informal}) makes the proof of \Cref{lem:all-fnexp-informal} particularly simple.
Suppose that the first item fails, in which case there exists a classical language $L$ for which $\calU \in \uP^L$.
We construct an $\NEXP$ machine that expects a witness $w$ interpreted as an assignment to $L(x)$ on all strings $x$ of polynomial length.
Given $w$, the $\NEXP$ verifier simulates the $\uP^L$ algorithm in exponential time by brute-force matrix multiplication, replacing oracle queries according to $w$.
The verifier uses \Cref{thm:entrywise_characterization_informal} to compute $\calU$ entrywise and compares to the unitary implemented by the $\uP$ oracle machine, accepting if and only if they are sufficiently close.
Then, any witness $w$ that causes the $\NEXP$ machine to accept allows the $\uP$ machine to implement $\calU$, given oracle access to $w$.
This is to say that $\calU \in \uP^\FNEXP$, where $\FNEXP$ denotes $\NEXP$ witness search.

The remaining results about $\uPSPACE$ follow by combining \Cref{lem:all-fnexp-informal} with known facts from classical complexity theory.
For example, \Cref{cor:am_nexp_informal} follows immediately from the easy witness lemma of Impagliazzo, Kabanets, and Wigderson~\cite{IKW02-easy-witness}, which shows that if $\AM = \NEXP$, then bits of accepting $\NEXP$ witnesses can be computed in $\EXP = \PSPACE$.
\Cref{cor:nc_np_informal} follows from a similar statement that if $\NC = \NP$, then bits of accepting $\NEXP$ witnesses can be computed in $\EXP = \PSPACE$.
Indeed, this is because $\PSPACE$ and $\NEXP$ are essentially just ``exponentially scaled up'' versions of $\NC$ and $\NP$, respectively.
Though, what is noteworthy in \Cref{cor:nc_np_informal} is that the equivalence of $\NC$ and $\NP$ leads to a collapse of $\NEXP$ to $\PSPACE$ \textit{relative to any oracle}, as long as the $\NEXP$ queries are only on strings of polynomial length.
This is why \Cref{cor:am_nexp_informal,cor:nc_np_informal} relativize differently.

Finally, the connection to quantum commitments (\Cref{thm:main-informal}) results from proving that commitments cannot exist if $\uP = \uPSPACE$.
This was essentially established in \cite{BEMPQY26-uhlmann}, though we reprove it ourselves to ensure consistency with our definitions.
We use the fact that commitment schemes are equivalent to the existence of an \textit{EFI pair}: a pair of efficiently-generated density matrices $\rho_0$, $\rho_1$ that are far in trace distance, but computationally indistinguishable to polynomial-time quantum adversaries~\cite{BCQ23-efi,Yan22-commitments}.
Using the entrywise characterization of $\uPSPACE$, we show that $\uPSPACE$ can approximately implement the statistically optimal measurement to distinguish $\rho_0$ and $\rho_1$.
Consequently, $\uP = \uPSPACE$ would violate the security of the EFI pair against polynomial time adversaries.

\subsection{Discussion}

This work is heavy on conceptual contributions.
To name just a few, we show or argue that:

\begin{enumerate}
    \item Proving the existence of computationally-secure quantum cryptography is as hard as longstanding open problems in complexity theory.
    \item The unitary synthesis problem explicitly relates unitary complexity theory, classical complexity theory, and cryptography.
    \item Sensible definitions of unitary complexity classes require polylogarithmic error dependence.
    \item Unitary complexity theory does not have to be ugly; our definitions yield many elegant proofs.
    \item Notions of unitary implementation with and without garbage are each interesting and meaningfully related to one another.
    \item $\uPSPACE$ is a particularly robust class, having a simple entrywise characterization and generically allowing garbage erasure.
\end{enumerate}

In the era of AI-automated theory research, we believe these insights are comparatively valuable, as today's LLMs have superhuman capabilities in proving theorems but still struggle to extract big ideas or meaning.

The most concrete and natural problem left open by our work is whether polynomial-time unitary implementations can always erase garbage:

\begin{question}[Garbage erasure]
    \label{q}
    For every $\calU \in \puP$, is there a $\calV \in \uP$ that equals $\calU$ up to phase?
\end{question}

Versions of \Cref{q} have been raised before, implicitly by Aaronson~\cite[Question 3.3.2]{Aar16-barbados} and more explicitly by Rosenthal~\cite[Open Problem 1]{Ros23-thesis}, but to our knowledge there has been no meaningful progress toward answering \Cref{q}.
We showed a negative answer to \Cref{q} relative to a unitary oracle (\Cref{thm:no_global_phase_informal}), but we do not consider this particularly strong evidence in favor of a negative answer to the general question, because we are not even aware of a \textit{candidate} unitary language that separates the two classes in the unrelativized setting.

For a more open-ended direction, we consider it interesting to explore the other possible consequences of unitary synthesis for cryptography.
To give one example, we showed that a positive solution to unitary synthesis is sufficient for quantum cryptography to imply breakthrough classical lower bounds.
But could a \textit{negative} solution to unitary synthesis also cause quantum cryptography to yield classical circuit lower bounds?
In particular, suppose we had an EFI pair that were an ``explicit'' counterexample to the unitary synthesis conjecture, in the sense that any oracle-assisted quantum algorithm that breaks the cryptosystem must have superpolynomial complexity.
Would the existence of such an ``explicit'' hard EFI pair \textit{itself} have important classical consequences?

\section{Preliminaries}

\subsection{Quantum circuits}
We begin by defining a notion of approximation that will be used extensively in this work.
This definition borrows from Kuperberg~\cite{Kup15-jones}.

\begin{definition}[FPTEAS/FPSEAS]
    \label{def:fpteas}
    An encoding function $f: \{0,1\}^* \to \R$ has a \emph{fully polynomial time exponential approximation scheme}, or \textit{FPTEAS} for short, if there is a polynomial-time Turing machine $M$ whose output $\phi = M(x, \eps)$ satisfies $\abs{\phi - f(x)} \le \eps$, where $\phi$ and $\eps$ are encoded in fixed-point binary.\footnote{We will always assume $\eps = 2^{-\lambda}$ for some $\lambda \in \N$ without loss of generality, as for larger inputs one can run the algorithm with $\eps = 1$.}

    A \emph{fully polynomial space exponential approximation scheme} (\textit{FPSEAS}) is defined the same way, but where $M$ is a polynomial-\underline{space} Turing machine, and the output has polynomial length.
\end{definition}

Note that the time or space bound on $M$ is $\poly(|x|, \log 1/\eps)$.
Some examples of numbers that admit an FPTEAS encoding are all numbers that have a fixed-point binary representation (i.e., dyadic rationals) and all algebraic numbers.
If we are willing to modify the encoding slightly, it is also without loss of generality to take $M$ to be any universal Turing machine.

Also following Kuperberg~\cite{Kup15-jones}, we use FPTEAS encodings to specify gates in a quantum circuit.

\begin{definition}[Quantum circuit]
    \label{def:quantum_circuit}
    A \emph{quantum circuit} is a finite sequence of gates acting on qubits, in which each gate is specified by
    \begin{enumerate}
        \item the list of $k$ qubits it acts on, and
        \item a CPTP map on $2^k \times 2^k$ density matrices, described by an FPTEAS encoding $f$ of its matrix representation on the operator basis $\{\ketbra{i}{j}\}_{i, j \in [2^k]}$.
        That is, a channel $\Lambda$ is specified by the matrix $S$ with entries $S_{ij,k\ell} = \bra{i}\Lambda(\ketbra{k}{\ell})\ket{j}$, and each complex entry $a + bi$ of $S$ is described by a pair $x, y \in \{0,1\}^*$ for which $f(x) = a$ and $f(y) = b$.
    \end{enumerate}
\end{definition}

The description length of a circuit is the total length of the gate descriptions together with the qubit labels on which the gates act.
We treat each gate description as specifying an exact channel, even though its entries may only be accessed computationally through approximations.

Allowing entries with an FPTEAS means that we may exactly specify gates such as rotations by rational multiples of $2\pi$ and Fourier transforms.
Our definition abstracts away the details of compiling gates into a universal gate set and leaves it up to the algorithm to specify the precision of each gate, if applicable.

\begin{definition}[Unitary quantum circuit]
    \label{def:unitary_circuit}
    A quantum circuit is \textit{unitary} if each gate is a unitary superoperator $\rho \to U\rho U^\dagger$ specified by the $2^k \times 2^k$ unitary matrix $U$ (with explicit global phase).
\end{definition}

Every $n$-qubit unitary circuit defines a unitary matrix in $\C^{2^n \times 2^n}$ (with explicit global phase) obtained by taking the product of the unitaries.

We also define classical oracle-augmented quantum circuits in the standard way:

\begin{definition}[Oracle quantum circuit]
    For a language $L$, an \emph{$L$-oracle quantum circuit} is a sequence that may contain both explicitly specified $k$-qubit gates (as above) and $n$-qubit oracle gates $U_n$ for the language $L$, which act unitarily as
    \[
    U_n\ket{x, b} = \ket{x, b \oplus L(x)}
    \]
    for every $x \in \{0,1\}^{n-1}$ and $b \in \{0,1\}$.
    These oracle gates are specified only by the qubits acted upon.
\end{definition}

\subsection{Quantum complexity theory}

\begin{definition}[Uniform circuit family]
    A family $\{C_x\}_{x \in \{0,1\}^*}$ of quantum circuits is \textit{time-uniform} if there exists a polynomial-time Turing machine $M$ with a write-only output tape such that $M(x)$ outputs a description of $C_x$.

    A \textit{space-uniform} family is defined the same way, but with a polynomial-\underline{space} machine $M$, and with the restriction that each gate has a polynomial-length description.\footnote{For example, we do not want to allow the space-uniform machine to output a gate whose FPTEAS encoding is exponentially long.
    We want each gate in the circuit to be storable in polynomial memory.}
\end{definition}

When a circuit family is indexed by several parameters, such as $\{C_{n, \eps}\}_{n \in \N, \eps > 0}$, we regard it as a family indexed by a single binary string encoding those parameters.
In such circuit families, we will always encode the parameter labeled $n \in \N$ in unary (i.e., $1^n$), and we will always encode the parameter labeled $\eps > 0$ as a fixed-point binary number of length $O(\log 1/\eps)$.

\begin{definition}[$\BQPSPACE$/$\BQuPSPACE$]
    \label{def:bqpspace}
    A promise problem $\Pi$ is in $\BQPSPACE^{L}$ if there exists a space-uniform family of $L$-oracle quantum circuits $\{C_n\}_{n \in \N}$ on $n + a(n) \le \poly(n)$ qubits such that for every $x \in \{0,1\}^n$ that is in the promise of $\Pi$, measuring a designated output qubit of $C_n$ applied to $\ket{x}\ket{0^a}$ yields $\Pi(x)$ with probability at least $2/3$.

    $\BQuPSPACE^{L}$ is defined identically, except with a space-uniform family of \underline{unitary} quantum circuits.
\end{definition}

\begin{remark}
    Some definitions of space-bounded quantum computation use circuit families indexed by the input $x$ itself, rather than only by its length $n$, e.g.~\cite{FR21-space}.
    This does not change the resulting class because the circuit $C_n$ can simulate the computation of $M(x)$ internally (and in the case of $\BQuPSPACE$, reversibly~\cite{LMP00-rpspace}).
    We use the length-uniform formulation because it is better suited for treating algorithms as coherent subroutines that can be applied to a superposition of inputs.
\end{remark}

It is well-known that $\BQPSPACE = \BQuPSPACE = \PSPACE$~\cite{Wat99-space,Wat03-space,FR21-space}.
Importantly for us, this proof also relativizes, a fact which could have been observed at any point in the last 25 years but apparently never was.\footnote{As recently as last year, \cite{AK25-oracles} wrote ``the proof of this result might not relativize for classical oracles according to our understanding''.
To our knowledge, the only subtlety in whether the proof relativizes or not is that there are alternative definitions of $\PSPACE$ in which an exponentially-long oracle tape is allowed.}

\begin{proposition}
    \label{prop:bqpspace}
    For all languages $L$, $\BQPSPACE^L = \BQuPSPACE^L = \PSPACE^L$.
\end{proposition}

\begin{proof}
    The containment $\PSPACE^L \subseteq \BQuPSPACE^L$ follows from the simulation of $\PSPACE$ by \textit{reversible} polynomial space (and therefore unitary quantum space)~\cite{LMP00-rpspace}.
    The containment $\BQuPSPACE^L \subseteq \BQPSPACE^L$ is trivial.
    To complete the proof, we turn to showing $\BQPSPACE^L \subseteq \PSPACE^L$.

    Given $\Pi \in \BQPSPACE^L$, let $C_n$ be the $n + a \eqqcolon m$-qubit $L$-oracle quantum circuit that evaluates $\Pi$ in the sense of \Cref{def:bqpspace}.
    Let $\Lambda$ be the channel computed by $C$, and let $S$ be the matrix representation $S_{ij,k\ell} = \bra{i}\Lambda(\ketbra{k}{\ell})\ket{j}$.
    Notice that it suffices to devise a $\PSPACE^L$ routine to approximate entries of $S$, because the acceptance probability of $C_n$ on input $x$ is a sum of entries of $S$:
    \[
    \Pr[1] = \sum_{y \in \{0,1\}^{m-1}} \bra{1y}\Lambda(\ketbra{x0^a}{x0^a})\ket{1y}.
    \]
    
    The idea is to design a divide and conquer strategy to compute the entries of $S$.
    The key observation is that if $\Lambda = \Lambda^{(1)} \circ \Lambda^{(2)}$ is a composition of two channels, then the corresponding matrix representations satisfy $S = S^{(1)}S^{(2)}$.
    But multiplication of two implicitly-given matrices is a sum of products, which can be performed in $\PSPACE$:
    \[
    S_{\alpha,\gamma} = \sum_\beta S^{(1)}_{\alpha,\beta}S^{(2)}_{\beta,\gamma}
    \]
    In particular, if there is an algorithm that computes entries of $S^{(1)}$ and $S^{(2)}$ to error $\eps$ in space $T$, it yields an algorithm to compute entries of $S$ to error $O(4^m \eps)$ with space $T + O(m + \log 1/\eps)$.
    This is because the entries of $S$ have magnitude bounded by $1$, so the algorithm just needs to store a $2m$-bit counter corresponding to $\beta$ and record the cumulative sum to $O(\log 1/\eps)$ bits of precision.

    Note the entries of the oracle gates are accessible in $O(m)$ space by querying $L$.
    Also, the matrix representations of the non-oracle gates comprising $C_n$ are given explicitly with an FPTEAS, so their entries can be computed to error $\eps$ in space $\poly(n + \log 1/\eps)$.
    Assuming $C_n$ contains $M \le 2^{\poly(n)}$ gates, we recursively divide $C_n$ in half to compute entries of $S$ in the fashion described in the previous paragraph.
    The space required is the solution to the recurrence:
    \begin{align*}
    T(1, \eps) &= O(m)+  \poly(n, \log 1/\eps),\\
    T(M, \eps) &= T\left(\frac{M}{2}, \frac{\eps}{O(4^m)}\right) + O(m + \log 1/\eps).
    \end{align*}
    which satisfies
    \[
    T(M, \eps) = O\left(\log M \log 1/\eps + m (\log M)^2\right) + \poly\left(n + \log 1/\eps + m \log M\right).
    \]
    Set $\eps = 2^{-(m + 2)}$, so that the total error in summing over all $y \in \{0,1\}^{m-1}$ is at most $1/8$.
    This lets us distinguish whether $\Pr[1] \ge 2/3$ or $\Pr[1] \le 1/3$.
    The total space required is at most $\poly(n)$, which allows a $\PSPACE^L$ machine to decide $\Pi$.
\end{proof}

We conclude this section with our definition of unitary languages:

\begin{definition}[Unitary language] \label{def:u-language}
    A \emph{unitary language} is a collection of unitaries $\calU = \braces{U_n}_{n \in \N}$ with $U_n \in \C^{2^n \times 2^n}$.
\end{definition}

Throughout, $\uALL$ denotes the set of all unitary languages.
Unitary languages subsume ordinary languages via the standard XOR oracle encoding:

\begin{definition}[Unitary embedding]
    If $L: \{0,1\}^* \to \{0,1\}$ is a language, the \emph{unitary embedding} of $L$ is the unitary language $\calU = \braces{U_n}_{n \in \N}$ whose constituent unitaries act as
    \[
    U_n\ket{x, b} = \ket{x, b \oplus L(x)}
    \]
    for every $x \in \{0,1\}^{n-1}$ and $b \in \{0,1\}$.
\end{definition}

In a slight abuse of notation, we may sometimes overload $L$ to mean both the language and its unitary embedding, particularly when $L$ is accessed as an oracle.

\section{Unitary complexity definitions}

\subsection{Garbage-free unitary complexity}
\label{sec:tidy_definitions}

Our definition of unitary complexity classes with no garbage is based around Zhandry's notion of an \textit{extension}~\cite{Zha25-unitary-oracles}, which allows for implementing a unitary with ancilla that must be reset to zero.

\begin{definition}[Extension]
    \label{def:extension}
    If $U$ is an $n$-qubit unitary, we say that an $(n+a)$-qubit unitary $\overline{U}$ is an \emph{$a$-qubit extension of $U$} if, for every $n$-qubit $\ket{\psi}$, $\overline{U}\ket{\psi}\ket{0^a} = (U\ket{\psi})\ket{0^a}$.
\end{definition}

Because an extension $\overline{U}$ leaves the subspace of the form $\ket{\psi}\ket{0^a}$ invariant, it must decompose as a direct sum $\overline{U} = U \oplus V$ for some unitary $V$ of dimension $2^{n+a} - 2^n$ acting on the orthogonal complement (i.e., the span of states where the ancilla register is not all zero).
And equivalently, any unitary direct sum $U \oplus V$ across these two subspaces is an extension of $U$.

\begin{definition}[$\uP/\uPSPACE$]
\label{def:unitaryP_tidy}
    A unitary language $\calU = \{U_n\}_{n \in \N}$ is in $\uP$ if there is a time-uniform family of unitary quantum circuits $\{C_{n, \eps}\}_{n \in \N, \eps > 0}$ on $n + a(n, \eps) \le \poly(n, \log 1/\eps)$ qubits such that for every $n$ and $\eps$, there is an $a$-qubit extension $\overline{U}$ of $U_n$ satisfying $\opnorm{\overline{U} - C_{n, \eps}} \le \eps$.

    $\uPSPACE$ is defined the same way, but where $\{C_{n, \eps}\}_{n \in \N, \eps > 0}$ is a \underline{space}-uniform family.
\end{definition}

Since the circuit family is time-uniform (respectively, space-uniform), the length of $C_{n,\eps}$ and the time required to output it are at most $\poly(n, \log 1/\eps)$ (respectively, $2^{\poly(n, \log 1/\eps)}$).

On the surface, \Cref{def:unitaryP_tidy} might seem cumbersome or even unnatural because we only care about the correctness of $C_{n, \eps}$ on the subspace of the form $\ket{\psi}\ket{0^a}$, yet we have to argue the existence of this $\overline{U}$ that acts on the \textit{entire} $(n+a)$-qubit Hilbert space.
But there are two reasons to like this definition:
\begin{enumerate}
    \item The operator norm is useful for composability, because it behaves well under tensor product and multiplication.
    In particular, $\opnorm{A} = \opnorm{A \otimes I}$, and for unitary operators, $\opnorm{AB - A'B'} \le \opnorm{A - A'} + \opnorm{B - B'}$.
    These properties will often be useful when we have a circuit that uses $C_{n, \eps}$ as a subroutine, in which case we can argue that the overall action of the circuit does not change much if we replace $C_{n, \eps}$ by $\overline{U}$.
    \item In fact, we can deduce the existence of a close extension $\overline{U}$ solely by arguing that $C_{n, \eps}$ approximately implements $U$ on the zero ancilla subspace.
    The next lemma establishes this.
\end{enumerate}

\begin{lemma}[Approximate implementations are close to extensions]
\label{lem:extension_close_op_norm}
    Let $C$ be an $(n+a)$-qubit unitary matrix such that, for all $n$-qubit states $\ket{\psi}$, $\norm{C \ket{\psi} \ket{0^a} - (U \ket{\psi}) \ket{0^a}} \leq \eps / 2$.
    Then there exists an $a$-qubit extension $\overline{U}$ of $U$ satisfying $\opnorm{C - \overline{U}} \leq \eps$.
\end{lemma}
\begin{proof}
When there are no ancilla ($a = 0$), this statement is trivial.
Suppose that $a \geq 1$.
We use the cosine-sine decomposition~\cite{PW94-cosine-sine}.
We consider $C$ as having two blocks, corresponding to the rows $\bra{x}\bra{0^a}$ and the rest, $\bra{x}\bra{y}$ for $y \neq 0^a$.
Note that the first block is the smaller one.
Then the cosine-sine decomposition states that we can write
\begin{align*}
    C = \begin{bmatrix}
        V_1 \\ & V_2
    \end{bmatrix} \begin{bmatrix}
        S_1 & S_2^\dagger \\ S_2 & -S_1 \oplus I
    \end{bmatrix} \begin{bmatrix}
        W_1^\dagger \\ & W_2^\dagger
    \end{bmatrix}
\end{align*}
where $V_1$, $V_2$, $W_1$, and $W_2$ are unitary; and $S_1$ and $S_2$ are real diagonal, with $S_1$ being square and $S_2$ being rectangular; and $S_1 \oplus I$ denotes $S_1$, direct sum with an identity matrix with the size needed to make it match the expected size of the block.
The initial assumption corresponds to the left two blocks of $C$ being close to $U$ and $0$, respectively:
\begin{align*}
    \opnorm[\Big]{
    \begin{bmatrix}
        V_1 S_1 W_1^\dagger \\ V_2 S_2 W_1^\dagger
    \end{bmatrix} - \begin{bmatrix}
        U \\ 0
    \end{bmatrix}
    }
    = \opnorm[\Big]{
    \begin{bmatrix}
        S_1 - V_1^\dagger U W_1 \\ S_2
    \end{bmatrix}
    } \leq \eps / 2.
\end{align*}
In particular, this means that $\opnorm{S_1 - V_1^\dagger U W_1} \leq \eps / 2$ and $\opnorm{S_2} \leq \eps / 2$.
Now let
\[
\overline{U} \coloneqq U \oplus (V_2(-V_1^\dagger U W_1 \oplus I) W_2^\dagger).
\]
Then this choice of $\overline{U}$ satisfies the desired bound:
\begin{align*}
    \opnorm{C - \overline{U}}
    &= \opnorm[\Big]{\begin{bmatrix}
        V_1 \\ & V_2
    \end{bmatrix} \begin{bmatrix}
        S_1 & S_2^\dagger \\ S_2 & -S_1 \oplus I
    \end{bmatrix} \begin{bmatrix}
        W_1^\dagger \\ & W_2^\dagger
    \end{bmatrix}
    - \begin{bmatrix}
        U \\ & V_2 (-V_1^\dagger U W_1 \oplus I) W_2^\dagger
    \end{bmatrix}} \\
    &= \opnorm[\Big]{\begin{bmatrix}
        S_1 - V_1^\dagger U W_1 & S_2^\dagger \\ S_2 & (-S_1 + V_1^\dagger U W_1) \oplus 0
    \end{bmatrix}} \\
    &\leq 2 \max\Big(\opnorm{S_1 - V_1^\dagger U W_1}, \opnorm{S_2}\Big) \\
    &\leq \eps.
\end{align*}
The final step uses our assumption.
\end{proof}

Much like Zhandry~\cite{Zha25-unitary-oracles}, we observe that one can efficiently compute the conjugate, transpose, inverse, or controlled version of an efficiently-computable unitary.
Our definition makes the proof somewhat simpler, because we do not need to use the Solovay--Kitaev theorem to recompile gates; e.g.\ we can directly compute the inverse of a gate by acting on the FPTEASs of the entries.

\begin{proposition}[Unitary homomorphisms]
    \label{prop:homomorphisms_in_p}
    Let $\class{K} \in \{\P, \PSPACE\}$ and $\calU \in \class{unitaryK}$. Then
    \[
    \calU^*, \calU^\dagger, \calU^\trans, \controlled\calU, \controlled\calU^*, \controlled\calU^\dagger, \controlled\calU^\trans \in \class{unitaryK}.
    \]
\end{proposition}

\begin{proof}
    We can modify gate-by-gate the circuit $C_{n,\eps}$ that computes $U_n$ to error $\eps$ in the sense of \Cref{def:unitaryP_tidy}.
    For example, in the case of $\controlled\calU^*$, replace each gate with its controlled complex conjugate.
    This implements the controlled conjugate of $C_{n,\eps}$, which satisfies $\opnorm{\controlled\overline{U}^* - \controlled C_{n,\eps}^*} = \opnorm{\overline{U} - C_{n,\eps}} \le \eps$ for the desired extension $\overline{U}$ of $U_n$.
    The same argument works for the other transformations by taking the gatewise conjugate, inverse, transpose, or controlled version thereof, and reversing the order in the case of the inverse or transpose.
\end{proof}

The ability to compute such homomorphisms efficiently motivates giving access to all of the homomorphisms when accessing a unitary language as an oracle.

\begin{definition}[Unitary oracle circuit]
    \label{def:oracle_circuit}
    Let $\calO = \braces{O_n}_{n \in \N}$ be a unitary language.
    An \emph{$\calO$-oracle unitary quantum circuit} is a sequence that may contain both explicitly specified $k$-qubit unitary gates (as in \Cref{def:unitary_circuit})  and oracle gates from the set
    \[
    \braces{O_n, O^*_n, O^\dagger_n, O^\trans_n, \controlled O_n, \controlled O_n^*, \controlled O_n^\dagger, \controlled O_n^\trans}_{n \in \N}.
    \]
    These oracle gates are specified only by the label and the qubits acted upon.
\end{definition}

\begin{definition}[Oracle $\uP/\uPSPACE$]
    \label{def:unitaryP_tidy_oracle}
    A unitary language $\calU = \{U_n\}_{n \in \N}$ is in $\uP^\calO$ if there is a time-uniform family of \underline{$\calO$-oracle} unitary quantum circuits $\{C_{n, \eps}\}_{n \in \N, \eps > 0}$ on $n + a(n, \eps) \le \poly(n, \log 1/\eps)$ qubits such that for every $n$ and $\eps$, there is an $a$-qubit extension $\overline{U}$ of $U_n$ satisfying $\opnorm{\overline{U} - C_{n, \eps}} \le \eps$.

    $\uPSPACE^\calO$ is defined the same way, but where $\{C_{n, \eps}\}_{n \in \N, \eps > 0}$ is a \underline{space}-uniform family.
\end{definition}

Here we prove that our definition of unitary complexity satisfies some basic desiderata.
For the first proof, that unitary languages subsume ordinary quantum-computable languages, we require the following lemma:

\begin{lemma}[Unitary column approximation yields global approximation]
    \label{lem:basis_to_op_norm}
    Suppose that for every $x \in \{0,1\}^n$,
    \[
    \norm{C\ket{x}\ket{0^a} - (U\ket{x})\ket{0^a}} \le \frac{\eps}{\sqrt{2^n}}.
    \]
    Then for every $n$-qubit $\ket{\psi}$,
    \[
    \norm{C\ket{\psi}\ket{0^a} - (U\ket{\psi})\ket{0^a}} \le \eps.
    \]
\end{lemma}

\begin{proof}
    Write $\ket{\psi} = \sum_{x \in \{0,1\}^n} \alpha_x \ket{x}$.
    Then
    \begin{align*}
    \norm{C\ket{\psi}\ket{0^a} - (U\ket{\psi})\ket{0^a}}
    &\le \sum_{x \in \{0,1\}^n} \norm{\alpha_x C\ket{x}\ket{0^a} - \alpha_x(U\ket{x})\ket{0^a}}\\
    &\le \frac{\eps}{\sqrt{2^n}} \sum_{x \in \{0,1\}^n} \abs{\alpha_x}\\
    &\le \eps,
    \end{align*}
    where the inequalities hold by the triangle inequality, the assumption of the lemma, and Cauchy-Schwarz, respectively.
\end{proof}

\begin{proposition}[Unitary classes subsume languages]
    \label{prop:tidy_contains_bbqp}
    Fix a unitary oracle $\calO$ and language $L$, and let $\calU$ be its unitary embedding.
    Then
    \[
    L \in \BQP^\calO \implies \calU \in \uP^\calO; \qquad L \in \BQuPSPACE^\calO \implies \calU \in \uPSPACE^\calO,
    \]
    where the unitary oracle model for the language classes allows the same gates as in \Cref{def:oracle_circuit}.
\end{proposition}

\begin{proof}
    Let $\class{K} \in \{\P, \class{_UPSPACE}\}$.
    Since $L \in \mathsf{BQK}^\calO$, there is a time or space (depending on $\class{K}$) uniform family of $\calO$-oracle unitary quantum circuits $\{A_n\}_{n \in \N}$ such that for any $x \in \{0,1\}^n$, measuring the first qubit of $A_n\ket{x}\ket{0^a}$ yields $L(x)$ with probability at least $2/3$.
    By parallel repetition, we can repeat $A_n$ $O(\log 1/\eps)$ times and majority vote to produce a circuit $A_{n, \eps}$ that succeeds with probability $1 - \eps$.
    Now let $\delta = \frac{\eps^2}{2^{n+5}}$ and let $C_{n, \eps}$ be the circuit that applies $A_{n-1,\delta}$, CNOTs the output bit into $\ket{b}$, and then applies $A_{n-1, \delta}^\dagger$.
    To analyze this construction, we first establish the following claim:

    \begin{claim}
        \label{claim:dist_bound}
        Suppose $\braket{\psi_0}{\psi_1} = 0$.
        Let $\ket{\phi} = \sqrt{1 - p}\ket{\psi_0} + \sqrt{p}\ket{\psi_1}$ for some $p \in [0, 1]$.
        Then
        \[
        \norm{\ket{\psi_0} - \ket{\phi}} \le \sqrt{2p}.
        \]
    \end{claim}

    \begin{proof}[Proof of claim]
    We have
    \begin{align*}
    \norm{\ket{\psi_0} - \ket{\phi}}^2
    &=
    \norm{
    (1 - \sqrt{1-p})\ket{\psi_0}
    -
    \sqrt{p}\ket{\psi_1}
    }^2 \\
    &=
    (1 - \sqrt{1-p})^2 + p\\
    &= 2 - 2\sqrt{1 - p}\\
    &\le 2p,
    \end{align*}
    where in the last line we apply the inequality $1 - p \le \sqrt{1 - p}$.
    Taking square roots completes the proof.
    \end{proof}

    Now consider any $x \in \{0,1\}^{n-1}$ and $b \in \{0,1\}$.
    There is some $p_{x,b} \le \delta$ for which applying $C_{n, \eps}$ to
    $\ket{x}\ket{b}\ket{0^a}$ yields
    \begin{align*}
    \ket{x}\ket{b}\ket{0^a} &\mapsto \sqrt{1 - p_{x,b}}\ket{x}\ket{b}\ket{L(x)}\ket{\phi_{x,b, L(x)}} + \sqrt{p_{x,b}}\ket{x}\ket{b}\ket{\lnot L(x)}\ket{\phi_{x,b, \lnot L(x)}}\\
    &\approx \ket{x}\ket{b}\ket{L(x)}\ket{\phi_{x,b, L(x)}}\\
    &\mapsto \ket{x}\ket{b \oplus L(x)}\ket{L(x)}\ket{\phi_{x,b, L(x)}}\\
    &\mapsto \sqrt{1 - p_{x,b}}\ket{x}\ket{b \oplus L(x)}\ket{0^a} + \sqrt{p_{x,b}} \ket{\bot}\\
    &\approx \ket{x}\ket{b \oplus L(x)}\ket{0^a},
    \end{align*}
    where both of the $\approx$ hold to within $\sqrt{2\delta}$ in Euclidean distance by the claim above.
    So, by the triangle inequality,
    \[
    \norm{C_{n, \eps}\ket{x}\ket{b}\ket{0^a} - \ket{x}\ket{b \oplus L(x)}\ket{0^a}} \le 2\sqrt{2\delta} = \frac{\eps}{2\sqrt{2^n}}.
    \]
    Appealing to \Cref{lem:basis_to_op_norm} and \Cref{lem:extension_close_op_norm} completes the proof, as $C_{n,\eps}$ is clearly time/space-uniform (depending on $\class{K}$).
\end{proof}

We next observe that our unitary complexity classes compose as expected:

\begin{proposition}[Unitary classes are self-low]
    \label{prop:tidy_low}
    Let $\class{K} \in \{\P, \PSPACE\}$. Then $\class{unitaryK}^{\class{unitaryK}^\calO} = \class{unitaryK}^\calO$.
\end{proposition}

\begin{proof}
    Suppose $\calU \in \class{unitaryK}^\calV$ for some $\calV \in \class{unitaryK}^\calO$.
    Let $C_{n,\eps/2}$ be the $\calV$-oracle unitary quantum circuit that solves $\calU$ to error $\eps/2$ on length-$n$ inputs.
    Let $m$ be an upper bound on the size of $C_{n,\eps/2}$, which is at most $\poly(n, \log 1/\eps)$ if $\class{K} = \P$ or $2^{\poly(n, \log 1/\eps)}$ if $\class{K} = \PSPACE$.
    To solve $\calU$ by a family of $\calO$-oracle unitary quantum circuits, define $C'_{n, \eps}$ to be the circuit obtained from $C_{n, \eps/2}$ by replacing each $\calV$-oracle gate with its error-$\eps/2m$ implementation, possibly at the cost of appending additional ancilla qubits initialized to $\ket{0}$ that will be reused at each oracle call.
    These implementations exist because of \Cref{prop:homomorphisms_in_p}.

    First recall from the definition of $\uP^\calV$ that $C_{n, \eps/2}$ is $\eps/2$-close in operator norm to an extension of $U$.
    That is, there exists a unitary $V$ such that
    \begin{equation}\label{eq:C_close_extension}
    \opnorm{C_{n, \eps/2} - U \oplus V} \le \eps / 2.
    \end{equation}

    Next, we claim that $C'_{n, \eps}$ is $\eps/2$ close to an extension of $C_{n, \eps/2}$.
    Or in symbols, there exists a unitary $W$ such that
    \begin{equation}\label{eq:C'_close_extension}
    \opnorm{C'_{n, \eps} - C_{n, \eps/2} \oplus W} \le \eps / 2.
    \end{equation}
    To see this, consider the circuit $\overline{C_{n, \eps/2}}$ that replaces each oracle query $Q_i$ in $C_{n, \eps/2}$ with the corresponding extension $\overline{Q}_i$ that is $\eps/2m$-close in operator norm to its implementation $Q'_i$ in $C'_{n, \eps}$.
    Then $\overline{C_{n, \eps/2}}$ is an extension of $C_{n, \eps/2}$, because a product of extensions is an extension of the product (and an extension tensored with identity is still an extension).
    Additionally, since $\opnorm{\overline{Q}_i - Q'_i} \le \eps/2m$ , replacing each $Q'_i$ in $C'_{n, \eps}$ with $\overline{Q}_i$ one at a time yields
    \[
    \opnorm{C'_{n, \eps} - \overline{C_{n, \eps/2}}} \le \sum_{i=1}^{m} \frac{\eps}{2m} = \frac{\eps}{2}.
    \]
    by the triangle inequality (and the preservation of operator norm under tensor product with identity).
    So, we pick $W$ such that $\overline{C_{n, \eps/2}} = C_{n, \eps/2} \oplus W$.

    Combining \Cref{eq:C_close_extension,eq:C'_close_extension} by the triangle inequality yields
    \begin{align*}
    \opnorm{C'_{n, \eps} - U \oplus V \oplus W}
    &\le \opnorm{C'_{n, \eps} - C_{n, \eps/2} \oplus W} + \opnorm{C_{n, \eps/2} \oplus 0 - U \oplus V \oplus 0}\\
    &= \opnorm{C'_{n, \eps} - C_{n, \eps/2} \oplus W} + \opnorm{C_{n, \eps/2} - U \oplus V}\\
    &\le \eps,
    \end{align*}
    which shows that $C'_{n, \eps}$ is $\eps$-close to an extension of $U$.
    Thus, we have established that $C'_{n,\eps}$ satisfies the distance criterion for $\calU$ in \Cref{def:unitaryP_tidy_oracle}.
    It is also clear that $C'_{n,\eps}$ is time/space-uniform (depending on $\class{K}$), because the logarithm of the inverse error of each $\calV$ gate is a polynomial: $\log(2m / \eps) \le \poly(n, \log 1/\eps)$.
    So, $\calU \in \class{unitaryK}^\calO$.
\end{proof}

\subsection{Unitary complexity allowing garbage}

Here we again build a series of unitary complexity class definitions, but this time we allow for generating garbage.
Compared to \Cref{sec:tidy_definitions}, the only difference is that \textit{extensions} will be replaced by \textit{dilations}, meaning: the implementation does not need to reset the ancilla register back to zero.

\begin{definition}[Dilation]
    \label{def:extension_garbage}
    If $U$ is an $n$-qubit unitary, we say that an $(n+a)$-qubit unitary $\overline{U}$ is an \emph{$a$-qubit dilation of $U$} if there exists an $a$-qubit garbage state $\ket{g}$ such that for every $n$-qubit $\ket{\psi}$, $\overline{U}\ket{\psi}\ket{0^a} = (U\ket{\psi})\ket{g}$.
\end{definition}

Dilations in this sense are precisely the unitaries which implement Stinespring dilations of the channel $\rho \mapsto U \rho U^\dagger$.
In subsequent definitions within this section, we will underline the key differences with the corresponding definitions in \Cref{sec:tidy_definitions}.

\begin{definition}[$\puP/\puPSPACE$]
\label{def:unitaryP_garbage}
    A unitary language $\calU = \{U_n\}_{n \in \N}$ is in $\puP$ if there is a time-uniform family of unitary quantum circuits $\{C_{n, \eps}\}_{n \in \N, \eps > 0}$ on $n + a(n, \eps) \le \poly(n, \log 1/\eps)$ qubits such that for every $n$ and $\eps$, there is an $a$-qubit \underline{dilation} $\overline{U}$ of $U_n$ satisfying $\opnorm{\overline{U} - C_{n, \eps}} \le \eps$.

    $\puPSPACE$ is defined the same way, but where $\{C_{n, \eps}\}_{n \in \N, \eps > 0}$ is a \underline{space}-uniform family.
\end{definition}

The prefix $\class{proj}$ stands for ``projective'' in the sense of the projective unitary group.
The name is justified by the observation that an implementation of a unitary $U$ with garbage state $\ket{g}$ is also an implementation of $\phi U$ with garbage state $\phi^*\ket{g}$, for any global phase $\phi$.
Put another way, dilations generically ``forget'' the global phase information.

One question arises in the definition of $\puP$ and $\puPSPACE$ that was irrelevant to $\uP$ and $\uPSPACE$, namely: why should we require implementation by a unitary quantum circuit if the algorithm is allowed to produce garbage?
It previously made sense to demand unitary implementation when the goal was to assign a global phase to every unitary, because information about global phase is lost under non-unitary maps.
But if we don't care about global phase and allow producing garbage, then why not consider a more general definition in which the implementation is a quantum circuit consisting of arbitrary channels?
In the case of $\puP$, we have an easy answer, because it is without loss of generality to choose a unitary implementation by dilating every gate into a unitary.
In particular, we can always replace each channel with a unitary purification at the cost of adding a number of qubits proportional to the number of gates in the circuit.
However, this strategy cannot work for $\puPSPACE$, because the circuit could contain exponentially many gates, and thus purifying the computation would require superpolynomial space!

We defer a formal resolution to this conundrum until later; see \Cref{rmk:pspace_purification}.
The gist is: it turns out that we can use the equivalence $\BQPSPACE = \BQuPSPACE = \PSPACE$ (\Cref{prop:bqpspace}) to purify any given implementation of a unitary, using at most a \textit{polynomial} overhead in space.
Thus, \Cref{def:unitaryP_garbage} would not become any more general if we were to define it allowing general channel gates.

By similar reasoning as before, we deduce that approximate implementations are close to dilations:

\begin{lemma}[Approximate implementations with garbage are close to dilations]
\label{lem:garbage_extension_close_op_norm}
    Let $C$ be an $(n+a)$-qubit unitary matrix and $\ket{g}$ an $a$-qubit state such that, for all $n$-qubit states $\ket{\psi}$, $\norm{C \ket{\psi} \ket{0^a} - (U \ket{\psi}) \ket{g}} \leq \eps / 2$.
    Then there exists an $a$-qubit dilation $\overline{U}$ of $U$ satisfying $\opnorm{C - \overline{U}} \leq \eps$.
\end{lemma}

\begin{proof}
    Let $V$ be a unitary satisfying $V\ket{g} = \ket{0^a}$.
    Then for every $n$-qubit $\ket{\psi}$,
    \[
    \norm{(I \otimes V)C \ket{\psi} \ket{0^a} - (U \ket{\psi}) \ket{0^a}} \leq \eps / 2,
    \]
    and thus by \Cref{lem:extension_close_op_norm}, there is an extension $\widetilde{U}$ of $U$ satisfying $\opnorm{(I \otimes V)C - \widetilde{U}} \le \eps$.
    Equivalently, we have $\opnorm{C - (I \otimes V^\dagger)\widetilde{U}} \le \eps$.
    Noticing that $\overline{U} \coloneqq (I \otimes V^\dagger)\widetilde{U}$ is a dilation of $U$, we are done.
\end{proof}

We also deduce that approximation in the diamond norm suffices to obtain a close dilation:

\begin{lemma}[Approximation in diamond norm yields close dilations]
    For an $n$-qubit unitary $U$, let $\Phi: \rho \mapsto U \rho U^\dagger$ and let $C$ be an $(n + a)$-qubit unitary such that the map $\calC: \rho \mapsto \tr_2(C (\rho \otimes \proj{0^a}) C^\dagger)$ is $\eps$-close to $\Phi$ in diamond distance.
    Then there exists an $a$-qubit dilation $\overline{U}$ of $U$ such that $\opnorm{\overline{U} - C} \leq 4 \sqrt{\eps}$.
\end{lemma}
\begin{proof}
By the continuity of Stinespring isometries~\cite[Theorem 1]{vomEnde2023}, if $\dnorm{\Phi - \calC} \leq \eps$, then there is an $a$-qubit dilation $\overline{V}$ of $\Phi$ and an $a$-qubit dilation $\overline{W}$ of $\calC$ such that their ``isometry'' parts are close, $\opnorm{\overline{V}(I_n \otimes \ket{0^a}) - \overline{W}(I_n \otimes \ket{0^a})} \leq 2 \sqrt{\eps}$.
Any two dilations of a channel are related by a unitary on the dilating register, and $C$ is a dilation of $\calC$, so there is an $a$-qubit unitary $M$ such that $(I_n \otimes M) \overline{W} = C$.
Consequently, for $\overline{U} = (I_n \otimes M) \overline{V}$, we have that $\overline{U}\ket{\psi} = (U \ket{\psi})\ket{g}$ for some $\ket{g}$ (as all dilations of $\Phi$ take this form), and further, $\norm{\overline{U}\ket{\psi}\ket{0^a} - C\ket{\psi} \ket{0^a}} \leq 2 \sqrt{\eps}$.
So, using \cref{lem:garbage_extension_close_op_norm}, we can extend the isometry $\overline{U}$ into a unitary $\overline{\overline{U}}$ such that $\opnorm{\overline{\overline{U}} - C} \leq 4 \sqrt{\eps}$ as desired.
\end{proof}

As before, $\puP$ is closed under certain homomorphisms and anti-homomorphisms. 

\begin{proposition}[Projective unitary homomorphisms]
    \label{prop:homomorphisms_in_p_g}
    Let $\class{K} \in \{\P, \PSPACE\}$ and $\calU \in \class{projUnitaryK}$.
    Then $\calU^*, \calU^\dagger, \calU^\trans \in \class{projUnitaryK}$.
\end{proposition}

\begin{proof}
    For $\calU^*$, the proof is identical to \Cref{prop:homomorphisms_in_p}.
    
    The proof is subtler for $\calU^\dagger$, because the inverse of a dilation need not be a dilation of the inverse.
    The solution is to first run $C_{n, \eps/2}\ket{0^n}\ket{0^a}$ to obtain an approximate copy of the garbage state $\ket{g}$, and then run $C_{n, \eps/2}^\dagger \ket{\psi}\ket{g}$ to approximately obtain $(U_n^\dagger\ket{\psi})\ket{0^a}$.
    If we consider this operation with $C_{n, \eps/2}$ replaced by its $\eps/2$-close dilation, the resulting unitary maps $\ket{\psi}\ket{0^n}\ket{0^a}$ to $(U^\dagger\ket{\psi})(U\ket{0^n})\ket{0^a}$ under an appropriate ordering of the registers, and is thus a dilation of $U_n^\dagger$ with $n + a$ ancilla.
    Substituting $C_{n, \eps/2}$ twice introduces at most $\eps$ error in operator norm.

    Lastly, the case of $\calU^\trans$ follows from the previous two because $\calU^\trans = (\calU^\dagger)^*$.
\end{proof}

Note that we cannot deduce $\controlled \calU \in \class{projUnitaryK}$.
Intuitively, this is because the global phase of an implementation is undetermined, as the phase can be absorbed into the garbage state.
In fact, this is inherent: consider
\[
\calU \coloneqq \{(-1)^{HALT(n)}I^{\otimes n}\}_{n \in \N},
\]
where $HALT(n)$ is $1$ if the $n$th Turing machine halts and $0$ otherwise.
Then $\calU \in \puP$, by taking the identity circuit and letting $\ket{g} = (-1)^{HALT(n)}$ (a global phase, or $0$-qubit garbage state).
On the other hand, an efficient implementation of $\controlled\calU$ cannot exist because it would give an algorithm to decide the halting problem.

Analogous to \Cref{sec:tidy_definitions}, we define the oracle access mechanism for projective unitary complexity classes on the basis of the homomorphisms available to them.

\begin{definition}[Projective unitary oracle circuit]
    \label{def:projective_oracle_circuit}
    Let $\calO = \braces{O_n}_{n \in \N}$ be a unitary language.
    An \emph{projective $\calO$-oracle unitary quantum circuit} is a sequence that may contain both explicitly specified $k$-qubit unitary gates (as in \Cref{def:unitary_circuit})  and oracle gates from the set
    \[
    \braces{O_n, O^*_n, O^\dagger_n, O^\trans_n}_{n \in \N}.
    \]
    These oracle gates are specified only by the label and the qubits acted upon.
\end{definition}

\begin{definition}[Oracle $\puP/\puPSPACE$]
    A unitary language $\calU = \{U_n\}_{n \in \N}$ is in $\puP^\calO$ if there is a time-uniform family of \underline{projective $\calO$-oracle} unitary quantum circuits $\{C_{n, \eps}\}_{n \in \N, \eps > 0}$ on $n + a(n, \eps) \le \poly(n, \log 1/\eps)$ qubits such that for every $n$ and $\eps$, there is an $a$-qubit \underline{dilation} $\overline{U}$ of $U_n$ satisfying $\opnorm{\overline{U} - C_{n, \eps}} \le \eps$.

    $\puPSPACE^\calO$ is defined the same way, but where $\{C_{n, \eps}\}_{n \in \N, \eps > 0}$ is a \underline{space}-uniform family.
\end{definition}

Despite the differences in definition, we prove that our definition of unitary complexity with garbage satisfies the same desiderata as before.

\begin{proposition}[Projective unitary classes subsume languages]
    \label{prop:garbage_contains_bbqp}
    Fix a unitary oracle $\calO$ and language $L$, and let $\calU$ be its unitary embedding.
    Then
    \[
    L \in \BQP^\calO \implies \calU \in \puP^\calO; \qquad L \in \BQuPSPACE^\calO \implies \calU \in \puPSPACE^\calO,
    \]
    where the unitary oracle model for the language classes allows the same gates as in \Cref{def:projective_oracle_circuit}.
\end{proposition}

\begin{proof}
    Identical to \Cref{prop:tidy_contains_bbqp}, noting only two implicit differences that do not change the proof: projective oracle circuits do not allow controlled queries to $\calO$, and the garbage state is $\ket{g} = \ket{0^a}$.
\end{proof}

\begin{restatable}[Projective unitary classes are self-low]{proposition}{propgarbagelow}
    \label{prop:garbage_low}
    Let $\class{K} \in \{\P, \PSPACE\}$. Then $\class{projUnitaryK}^{\class{projUnitaryK}^\calO} = \class{projUnitaryK}^\calO$.
\end{restatable}

One could prove \Cref{prop:garbage_low} in a similar fashion to \Cref{prop:tidy_low} by replacing each of the oracle queries with their implementations and arguing that closeness to a dilation is preserved.
Such a direct proof becomes significantly more cumbersome than \Cref{prop:tidy_low}, however, because one must track the accumulation of garbage from each query.
Instead of trudging through such an ugly proof, we defer the proof of \Cref{prop:garbage_low} until later, as it will follow as a simple corollary of later results.

\subsection{Relations between definitions with and without garbage}

\label{sec:garbage_tidy_relations}

Clearly, the projective classes are strictly more general: $\uP \subseteq \puP$ and $\uPSPACE \subseteq \puPSPACE$, by setting $\ket{g} = \ket{0^a}$.

At first glance, one annoyance with the garbage definition (\Cref{def:unitaryP_garbage}) is that when querying an implementation of some $\calU \in \puP$, making many calls will produce garbage states that accumulate linearly with the number of queries.
This problem is particularly troubling when considering $\puPSPACE^\calU$, because a polynomial-space algorithm may query an oracle exponentially many times.
Thus, instantiating the oracle with a garbage implementation would seem to produce an exponential amount of garbage and therefore violate the space bound!

To avoid this issue, we first make use of a simple trick that reduces the garbage to at most $n$ qubits, without loss of generality.

\begin{theorem}
    \label{thm:u_garbage_iff_u_udag_tidy}
    Let $\class{K} \in \{\P, \PSPACE\}$.
    Then $\calU \in \class{projUnitaryK}$ if and only if $\calU \otimes \calU^\dagger \in \class{unitaryK}$.
\end{theorem}

\begin{proof}
    For the forward direction, fix $n$ and target error $\eps>0$.
    Let $C_{n,\eps/2}$ be the unitary quantum circuit on $n + a$ qubits that implements $U_n$ to error $\eps/2$ in the sense of \Cref{def:unitaryP_garbage}, and let $\overline{U}_n$ be the dilation of $U_n$ that is $\eps/2$-close to $C_{n,\eps/2}$ in operator norm.

    Consider the following circuit $C'_{2n, \eps}$ that takes in two $n$-qubit registers $\mathsf{X}$, $\mathsf{Y}$ and an $a$-qubit register $\mathsf{A}$:
    \[
    \Qcircuit @C=1em @R=1em {
    \lstick{\mathsf{X}} & \multigate{1}{C_{n, \eps/2}} & \qw & \qw\\
    \lstick{\mathsf{A}} & \ghost{C_{n, \eps/2}} & \multigate{1}{C^\dagger_{n, \eps/2}} & \qw\\
    \lstick{\mathsf{Y}} & \qw & \ghost{C^\dagger_{n, \eps/2}} & \qw
    }
    \]
    Now let $\overline{U}'_n$ be the unitary on $\mathsf{XYA}$ obtained by replacing the gates $C_{n, \eps/2}$ and $C^\dagger_{n, \eps/2}$ with $\overline{U}_n$ and $\overline{U}^\dagger_n$, respectively.
    Then for every $2n$-qubit $\ket{\psi}_\mathsf{XY}$ one has $\overline{U}'_n\ket{\psi}_\mathsf{XY}\ket{0^a}_\mathsf{A} = (U_{n}\otimes U^\dagger_{n})\ket{\psi}_\mathsf{XY}\ket{0^a}_\mathsf{A}$.
    Thus, $\overline{U}'_n$ is an extension of $U_n \otimes U_n^\dagger$. 
    Moreover, $\opnorm{\overline{U}'_n - C'_{2n, \eps}} \le \eps$, by the triangle inequality.
    Since $C'_{2n, \eps}$ is time/space-uniform (depending on $\class{K}$), we conclude that $\calU \otimes \calU^\dagger \in \class{unitaryK}$.
    
    The reverse direction is trivial, by setting $\ket{g} = (U_n^\dagger\ket{0^n})\ket{0^a}$.
\end{proof}

Interestingly, we can also show that the $\calU^\dagger$ in the previous theorem can be replaced by $\calU^*$:

\begin{theorem}
    \label{thm:u_garbage_iff_u_ustar_tidy}
    Let $\class{K} \in \{\P, \PSPACE\}$.
    Then $\calU \in \class{projUnitaryK}$ if and only if $\calU \otimes \calU^* \in \class{unitaryK}$.
\end{theorem}

\begin{proof}
    For the forward direction, fix $n$ and target error $\eps>0$.
    Let $C_{n,\eps/4}$ be the unitary quantum circuit on $n + a$ qubits that implements $U_n$ to error $\eps/4$ in the sense of  \Cref{def:unitaryP_garbage}, and let $\overline{U}_n$ be the dilation of $U_n$ that is $\eps/4$-close to $C_{n,\eps/4}$ in operator norm.
    
    Consider the following circuit $C'_{2n, \eps}$ that takes in $n$-qubit registers $\mathsf{X}$, $\mathsf{Y}$, $\mathsf{A}$, $\mathsf{B}$, and $a$-qubit registers $\mathsf{C}$, $\mathsf{D}$:

    \[
    \Qcircuit @C=1em @R=1em {
    \lstick{\mathsf{X}} & \multigate{1}{C} & \qw & \qw & \qw & \qw & \qw\\
    \lstick{\mathsf{C}} & \ghost{C} & \qw & \multigate{1}{C^\dagger} & \qw & \qw & \qw\\
    \lstick{\mathsf{A}} & \gate{H^{\otimes n}} & \ctrl{1} & \ghost{C^\dagger} & \ctrl{1} & \gate{H^{\otimes n}} & \qw \\
    \lstick{\mathsf{B}} & \qw & \targ & \multigate{1}{C^\trans} & \targ & \qw & \qw\\
    \lstick{\mathsf{D}} & \multigate{1}{C^*} & \qw & \ghost{C^\trans} & \qw & \qw & \qw\\
    \lstick{\mathsf{Y}} & \ghost{C^*} & \qw & \qw & \qw & \qw & \qw\\
    }
    \]
    Now let $\overline{U}'_n$ be the unitary on $\mathsf{XYABCD}$ obtained by replacing the gates $C_{n, \eps/4}$, $C^\dagger_{n, \eps/4}$, $C^*_{n, \eps/4}$, and $C^\trans_{n, \eps/4}$ with $\overline{U}_n$, $\overline{U}^\dagger_n$, $\overline{U}^*_n$, and $\overline{U}^\trans_n$, respectively.
    Then for every $2n$-qubit $\ket{\psi}_\mathsf{XY}$ one has $\overline{U}'_n\ket{\psi}_\mathsf{XY}\ket{0^{2n + 2a}}_\mathsf{ABCD} = (U_{n} \otimes U^*_{n})\ket{\psi}_\mathsf{XY}\ket{0^{2n + 2a}}_\mathsf{ABCD}$ (the idea is that the garbage gets written into $\mathsf{C}$ and $\mathsf{D}$ and later erased, crucially using the fact that the $2n$-qubit Bell state on $\mathsf{AB}$ is invariant under $U \otimes U^*$ for every unitary $U$).
    Thus, $\overline{U}'_n$ is an extension of $U_n \otimes U_n^*$.
    Moreover, $\opnorm{\overline{U}'_n - C'_{2n, \eps}} \le \eps$, by the triangle inequality.
    Since $C'_{2n,\eps}$ is time/space-uniform (depending on $\class{K}$), we conclude that $\calU \otimes \calU^* \in \class{unitaryK}$.

    The reverse direction is trivial, by setting $\ket{g} = (U_n^*\ket{0^n})\ket{0^a}$.
\end{proof}

We also have a reverse containment that characterizes $\class{unitaryK}$ in terms of $\class{projUnitaryK}$:

\begin{theorem}
    \label{thm:u_tidy_iff_cu_garbage}
    Let $\class{K} \in \{\P, \PSPACE\}$. 
    Then $\calU \in \class{unitaryK}$ if and only if $\controlled\calU \in \class{projUnitaryK}$.
\end{theorem}

\begin{proof}
    The forward direction follows from the fact that $\controlled\calU \in \class{unitaryK}$ (\Cref{prop:homomorphisms_in_p}) as well as the trivial containment $\class{unitaryK} \subseteq \class{projUnitaryK}$ (by setting $\ket{g} = \ket{0^a}$).

    For the reverse direction, fix $n$ and target error $\eps>0$.
    Let $C_{n,\eps/2}$ be the unitary quantum circuit on $n + 1 + a$ qubits that implements $\controlled U_n$ to error $\eps/2$ in the sense of  \Cref{def:unitaryP_garbage}, and let $\overline{U}_n$ be the dilation of $\controlled U_n$ that is $\eps/2$-close to $C_{n,\eps/2}$ in operator norm.
    
    Consider the following circuit $C'_{n, \eps}$ that takes in an $n$-qubit register $\mathsf{X}$, $1$-qubit register $\mathsf{Y}$, and $a$-qubit register $\mathsf{A}$:

    \[
    \Qcircuit @C=1em @R=1em {
    \lstick{\mathsf{X}} & \qw & \multigate{2}{C_{n, \eps/2}} & \qw & \multigate{2}{C_{n, \eps/2}^\dagger} & \qw\\
    \lstick{\mathsf{Y}} & \gate{X} & \ghost{C_{n, \eps/2}} & \gate{X} & \ghost{C_{n, \eps/2}^\dagger} & \qw\\
    \lstick{\mathsf{A}} & \qw & \ghost{C_{n, \eps/2}} & \qw & \ghost{C_{n, \eps/2}^\dagger} & \qw\\
    }
    \]
    Now let $\overline{U}'_n$ be the unitary on $\mathsf{XYA}$ obtained by replacing the gates $C_{n, \eps/2}$ and $C^\dagger_{n, \eps/2}$ with $\overline{U}_n$ and $\overline{U}^\dagger_n$, respectively.
    Then for every $n$-qubit $\ket{\psi}_\mathsf{X}$ one has $\overline{U}'_n\ket{\psi}_\mathsf{X}\ket{0^{a + 1}}_\mathsf{YA} = U_{n}\ket{\psi}_\mathsf{X}\ket{0^{a+1}}_\mathsf{YA}$ (the idea is that the garbage gets written into $\mathsf{A}$ and later erased, because when $\mathsf{Y}$ is $\ket{0}$, $\overline{U}^\dagger_n$ acts as the identity on $\mathsf{X}$).
    Thus, $\overline{U}'_n$ is an extension of $U_n$.
    Moreover, $\opnorm{\overline{U}'_n - C'_{n, \eps}} \le \eps$, by the triangle inequality.
    Since $C'_{n, \eps}$ is time/space-uniform (depending on $\class{K}$), we conclude that $\calU \in \class{unitaryK}$.
\end{proof}

Combining \Cref{thm:u_garbage_iff_u_udag_tidy,thm:u_garbage_iff_u_ustar_tidy,thm:u_tidy_iff_cu_garbage}, we see that $\class{unitaryK}$ and $\class{projUnitaryK}$ are both isomorphic to subclasses of each other.
This is not to say that the definitions are \textit{equivalent}; rather, any statement about one class has an equivalent reformulation in terms of the other.
In that sense, neither definition is more ``right'' than the other.

We can also prove suitable relativizations of the above statements:

\begin{theorem}
    \label{thm:u_garbage_iff_u_udag_tidy_oracle}
    Let $\class{K} \in \{\P, \PSPACE\}$.
    Then $\calU \in \class{projUnitaryK}^{\calO}$ if and only if $\calU \otimes \calU^\dagger \in \class{unitaryK}^{\calO \otimes \calO^\dagger}$.
\end{theorem}

\begin{proof}
    For the forward direction, start by following the proof of \Cref{thm:u_garbage_iff_u_udag_tidy}, but with projective oracle circuits instead of ordinary circuits.
    Augment $C'_{2n,\eps}$ with an additional register $\mathsf{Z}$ on $\poly(n, \log 1/\eps)$ qubits, and replace each $\calO_\mathsf{XAY}$-oracle gate in $C'_{2n,\eps}$ with a $\calO_\mathsf{XAY} \otimes \calO^\dagger_\mathsf{Z}$ gate.
    The new circuit is a projective ($\calO \otimes \calO^\dagger$)-oracle unitary quantum circuit, and it implements exactly $C'_{2n,\eps\mathsf{XAY}} \otimes I_\mathsf{Z}$ because the action on $\mathsf{Z}$ from $C_{n, \eps/2}$ will perfectly cancel with the action on $\mathsf{Z}$ from $C^\dagger_{n, \eps/2}$.
    Therefore it still implements $\calU \otimes \calU^\dagger$ but is now in $\class{unitaryK}^{\calO \otimes \calO^\dagger}$.

    The reverse direction is almost as trivial as in \Cref{thm:u_garbage_iff_u_udag_tidy}, with the one caveat that $(\calO \otimes \calO^\dagger)$-oracle circuits are not (obviously) \textit{projective} $\calO$-oracle circuits.
    In particular, we must show how to simulate controlled queries such as $\controlled(O \otimes O^\dagger)$ and $\controlled(O^* \otimes O^\trans)$ with ordinary queries to $O$, $O^\dagger$, $O^*$, and $O^\trans$.
    This is accomplished with the circuit below that implements $\controlled_\mathsf{A}(O_\mathsf{X} \otimes O^\dagger_\mathsf{Y})$:
    \[
    \Qcircuit @C=1em @R=1em {
    \lstick{\mathsf{A}} & \qw & \ctrlo{2} & \qw & \ctrlo{2} & \qw \\
    \lstick{\mathsf{X}} & \gate{O} & \qswap & \qw & \qswap & \qw\\
    \lstick{\mathsf{Y}} & \qw & \qswap & \gate{O^\dagger} & \qswap & \qw\\
    }
    \]
\end{proof}

By similarly extending the previous theorems to the oracle setting, one can establish the following:

\begin{theorem}
    Let $\class{K} \in \{\P, \PSPACE\}$.
    Then $\calU \in \class{projUnitaryK}^{\calO}$ if and only if $\calU \otimes \calU^* \in \class{unitaryK}^{\calO \otimes \calO^*}$.
\end{theorem}

\begin{theorem}
    Let $\class{K} \in \{\P, \PSPACE\}$.
    Then $\calU \in \class{unitaryK}^{\calO}$ if and only if $\controlled\calU \in \class{projUnitaryK}^{\controlled \calO}$.
\end{theorem}

We omit the proofs.
We conclude this section by noting that \Cref{thm:u_garbage_iff_u_udag_tidy_oracle} yields a particularly simple proof of \Cref{prop:garbage_low}:

\propgarbagelow*

\begin{proof}
    Suppose $\calU \in \class{projUnitaryK}^\calV$ for some $\calV \in \class{projUnitaryK}^\calO$.
    By \Cref{thm:u_garbage_iff_u_udag_tidy_oracle}, $\calU \otimes \calU^\dagger \in \class{unitaryK}^{\calV \otimes \calV^\dagger}$.
    Again by \Cref{thm:u_garbage_iff_u_udag_tidy_oracle}, $\calV \otimes \calV^\dagger \in \class{unitaryK}^{\calO \otimes \calO^\dagger}$, and thus $\calU \otimes \calU^\dagger \in \class{unitaryK}^{\class{unitaryK}^{\calO \otimes \calO^\dagger}}$.
    Appealing to \Cref{prop:tidy_low} puts $\calU \otimes \calU^\dagger \in \class{unitaryK}^{\calO \otimes \calO^\dagger}$.
    Finally, \Cref{thm:u_garbage_iff_u_udag_tidy_oracle} once more implies $\calU \in \class{projUnitaryK}^\calO$.
\end{proof}

\section{Characterization of \texorpdfstring{$\uPSPACE$}{unitaryPSPACE}}
\label{sec:characterization}

In this section, we give several different characterizations of the unitary languages in $\uPSPACE$, including one based on the computability of its entries, and another relating more closely to $\puPSPACE$.

\subsection{Approximation in polynomial space}

We first develop the theory of the real-valued functions that can be approximated by space-efficient quantum algorithms, and how such functions can be encoded into unitary languages.
This will serve as the basis for our later results that characterize $\uPSPACE$.
We first generalize the notion of an FPSEAS to functions that are approximable in polynomial \textit{quantum} space, perhaps with an oracle.

\begin{definition}[FPUSEAS]
    \label{def:FPUSEAS}
    An encoding function $f: \{0,1\}^* \to \R$ has a \emph{fully polynomial unitary space exponential approximation scheme}, or \textit{FPUSEAS} for short, if there is a family of space-uniform unitary quantum circuits $\{C_{n,\eps}\}_{n \in \N, \eps > 0}$ on $n + a(n, \eps) \le \poly(n, \log 1/\eps)$ qubits such that for every $x \in \{0,1\}^n$, measuring a designated output register of $C_{n,\eps}\ket{x}\ket{0^a}$ produces a fixed-point binary number $\phi$ satisfying $\abs{\phi - f(x)} \le \eps$ with probability at least $2/3$.

    If $\calO$ is a unitary language, we say that $f$ has an \emph{FPUSEAS relative to $\calO$} if $\{C_{n,\eps}\}_{n \in \N, \eps > 0}$ is instead a family of $\calO$-oracle circuits.
\end{definition}

\begin{remark}
By a Chernoff bound, the success probability in \Cref{def:FPUSEAS} can be amplified from $2/3$ to $1 - \delta$ by coherently repeating $O(\log 1/\delta)$ times and taking the median.
\end{remark}

We first show an equivalent characterization of FPSEAS approximation in terms of a related promise problem.

\begin{definition}[$f\Apx$]
    \label{def:f-apx}
    If $f: \{0,1\}^* \to \R$ is a function, then $f\Apx$ is the following promise problem defined on inputs $x \in \{0,1\}^*$, $\phi \in \R$, and $\eps > 0$, with the latter two encoded in fixed-point binary:
    \begin{itemize}
        \item (YES): $f\Apx(x, \phi, \eps) = 1$ if $f(x) > \phi + \eps$.
        \item (NO): $f\Apx(x, \phi, \eps) = 0$ if $f(x) < \phi - \eps$.
        \item Otherwise, $f\Apx(x, \phi, \eps)$ is undefined.
    \end{itemize}
\end{definition}

\begin{proposition}
    \label{prop:fapx_iff_fp}
    $f: \{0,1\}^* \to \R$ has an FPUSEAS relative to a unitary language $\calO$ if and only if $\abs{f(x)} \le 2^{\poly(|x|)}$ and $f\Apx \in \BQuPSPACE^\calO$.\footnote{$\BQuPSPACE^\calO$ is defined the same as in \Cref{def:bqpspace}, but with a family of $\calO$-oracle unitary quantum circuits.}
\end{proposition}

\begin{proof}
    For the forward implication, to solve an instance $(x, \phi, \eps)$ of $f\Apx$, run the FPUSEAS on $(x, \eps)$ to obtain $\phi'$ satisfying $\abs{\phi' - f(x)} \le \eps$.
    Output YES if $\phi' > \phi$ and NO otherwise.
    Then this is a solution to $f\Apx$ in $\BQuPSPACE^\calO$.
    Technically, the definition of $\BQuPSPACE^\calO$ requires that the space-uniform unitary circuit $C_m$ solving $f\Apx$ depends only on the length $m$ of the input $(x, \phi, \eps)$.
    However, the FPUSEAS specifies a circuit $C_{n, \eps}$ that depends on both the length of $x$ and the value of $\eps$.
    One solution is to let $C_m$ contain a copy of $C_{n, 2^{-m}}$ for each $n \le m$ after a decoding gadget that feeds $x$ into the correct sub-circuit.
    This works because fixed-point binary encoding requires $\eps \ge 2^{-m}$ and the output remains correct if we decrease the error parameter of the FPUSEAS.

    For the reverse implication, suppose $\abs{f(x)} \le 2^{\poly(|x|)}$ and $f\Apx \in \BQuPSPACE^\calO$.
    Then given input $x \in \{0,1\}^*$ and a desired precision $\eps > 0$, one can obtain an estimate $\phi$ satisfying $\abs{\phi - f(x)} \le \eps$ by binary search, using the $f\Apx$ algorithm as a subroutine.
    Specifically, we initialize the bounds $f(x) \in [\ell, u]$, where $\ell = -2^{\poly(|x|)}$ and $u = 2^{\poly(|x|)}$.
    At each step, we let $\mu = \frac{u + \ell}{2}$ and $\delta = \frac{u - \ell}{8}$ and query $f\Apx(x, \mu, \delta)$.
    If the output is YES, we ascertain from the promise that $f(x) \ge \mu - \delta$, while if the output is NO, we deduce $f(x) \le \mu + \delta$.
    In either case, we update the bounds $[\ell, u]$ on $f(x)$ accordingly.
    We repeat until $u - \ell \le 2\eps$ and finally output $\frac{u + \ell}{2}$.
    We implement this adaptive computation coherently, with deferred measurements.
    
    The length of the interval is multiplied by a factor of $5/8$ at each iteration, so the algorithm runs for at most $\poly(|x|) + O(\log 1/\eps)$ iterations before terminating.
    We set the failure probability of the $f\Apx$ subroutine to be inverse-linear in the number of iterations so that with probability at least $2/3$, the output of the subroutine is correct at all iterations.
    Assuming this happens, the algorithm is a valid FPUSEAS because it maintains the invariant $f(x) \in [\ell, u]$ whenever the $f\Apx$ subroutine responds appropriately.
    The total runtime overhead beyond the queries is at most $\poly(|x|, \log 1/\eps)$, so the overall algorithm has a space-uniform unitary family $\{C_{n, \eps}\}$.
\end{proof}

\begin{corollary}
    \label{cor:fpuseas_in_pspace}
    For any language $L$, $f: \{0,1\}^* \to \R$ has an FPUSEAS relative to $L$ if and only if $f$ has an FPSEAS relative to $L$.\footnote{An FPSEAS relative to $L$ takes the same definition as \Cref{def:fpteas}, except that $M$ may be an $L$-oracle Turing machine.}
\end{corollary}

\begin{proof}
    For the forward direction, an FPUSEAS for $f$ implies that $f\Apx \in \BQuPSPACE^L$ (\Cref{prop:fapx_iff_fp}), and therefore $f\Apx \in \PSPACE^L$ (\Cref{prop:bqpspace}).
    We can utilize the same binary search routine as in the proof of \Cref{prop:fapx_iff_fp} to obtain a classical polynomial-space approximation scheme with an oracle for $f\Apx$.

    The reverse direction follows from the observation that the output bits of the FPSEAS define a language in $\PSPACE^L$, which can be simulated in $\BQuPSPACE^L$ by \Cref{prop:bqpspace}.
\end{proof}

In the context of unitary complexity, it becomes natural to consider real-valued functions that are encoded into the entries of a unitary.

\begin{definition}[Phase encoding, $f\Phase$]
    \label{def:fphase}
    Let $f : \{0,1\}^* \to [-\pi/2, \pi/2]$.
    The \emph{phase encoding of $f$} is the unitary language $f\Phase = \{R_{f,n}\}_{n \in \N}$ defined by
    \[
    R_{f,n}\ket{x}
    =
    e^{i f(x)}\ket{x},
    \]
    for every $x \in \{0,1\}^n$.
\end{definition}

The restriction to $[-\pi/2, \pi/2]$ is merely to ensure the encoding is unique; one could straightforwardly generalize the definition to encode arbitrary elements of the circle group $\R / 2\pi\Z$.

\begin{lemma}
    \label{lem:rotation_to_fqpseas}
    $f\Phase \in \uPSPACE^{\calO}$ if and only if $f: \{0,1\}^* \to [-\pi/2, \pi/2]$ has an FPUSEAS relative to $\calO$.
\end{lemma}

\begin{proof}
    For the forward direction, suppose $f\Phase \in \uPSPACE^{\calO}$.
    We estimate $f(x)$ by applying phase estimation~\cite{kitaev95,NC10-book} to $R_{f,n}$ on input the eigenstate $\ket{x}$.
    Specifically, we choose the parameters of phase estimation so that the estimate $\phi$ satisfies $\abs{\phi - f(x)} \le \eps$ with probability at least $2/3$.
    Outputting $\phi$ gives an FPUSEAS \underline{relative to $f\Phase$}, as phase estimation uses a number of qubits polynomial in $n$ and $\log 1/\eps$.
    Now amplify the success probability of the FPUSEAS to at least $3/4$ (say) by parallel repetition.
    As $f\Phase \in \uPSPACE^{\calO}$, we can substitute the $f\Phase$ oracle gates with polynomial-space $\calO$-oracle unitary quantum circuits that approximate $f\Phase$ to small error, much like we did in \Cref{prop:tidy_low}.
    We choose the errors to be a small enough inverse-exponential in $n$ and $\log 1/\eps$ so that the overall error introduced is at most $1/12$.
    Then the resulting circuit gives an FPUSEAS relative to $\calO$, because the success probability is at least $3/4 - 1/12 = 2/3$.

    Conversely, suppose $f$ has an FPUSEAS relative to $\calO$.
    Fix $n$ and target error $\eps>0$.
    We may assume that there is a space-uniform family of $\calO$-oracle unitary circuits $\{A_{n, \eps}\}_{n \in \N, \eps > 0}$ such that, for every $x \in \{0,1\}^n$, measuring a designated output register of $A_{n, \eps}\ket{x}\ket{0^a}$ produces a number $\phi$ satisfying $\abs{\phi - f(x)} \le \gamma$ with probability at least $1-\delta$, where we choose
    \begin{equation}
    \label{eq:def_gamma_delta}
    \gamma = \frac{\eps}{4\sqrt{2^n}},
    \qquad
    \delta = \frac{\eps^2}{64\cdot 2^n}.
    \end{equation}

    Let $P$ be the unitary that applies the phase encoded by the output register.
    That is, we define $P$ by
    \[
    P\ket{x}\ket{\phi}\ket{w}
    =
    e^{i\phi}\ket{x}\ket{\phi}\ket{w},
    \]
    where $\ket{x}$ represents the input register of $A_{n, \eps}$, $\ket{\phi}$ represents the output, and $\ket{w}$ is the remainder of the workspace.

    Define $C_{n,\eps} \coloneqq A^\dagger_{n, \eps} P A_{n, \eps}$.
    We claim that $C_{n,\eps}$ approximately implements $R_{f,n}$ when applied to any fixed $x \in \{0,1\}^n$.
    Write
    \[
    \ket{\Gamma_x} \coloneqq A_{n, \eps}\ket{x}\ket{0^a}.
    \]
    Let $\Pi_x$ be the projector onto the ``good'' subspace in which the output register contains a value $\phi$ satisfying $\abs{\phi-f(x)} \le \gamma$.
    By assumption,
    \begin{equation}
    \label{eq:proj_diff_delta}
    \norm{(I-\Pi_x)\ket{\Gamma_x}} \le \sqrt{\delta}.
    \end{equation}
    Additionally, since $P$ is diagonal in the computational basis, we have the operator norm bound
    \begin{equation}
    \label{eq:opnorm_proj}
    \opnorm{(P - e^{if(x)}I)\Pi_x}
    =
    \sup_{\phi:\,\abs{\phi-f(x)}\le \gamma}
    \abs{e^{i\phi} - e^{if(x)}}
    \le \gamma.
    \end{equation}
    Therefore
    \begin{align*}
    \norm{P\ket{\Gamma_x} - e^{if(x)}\ket{\Gamma_x}}
    &\le
    \norm{(P-e^{if(x)}I)\Pi_x\ket{\Gamma_x}}
    +
    \norm{(P-e^{if(x)}I)(I-\Pi_x)\ket{\Gamma_x}}\\
    &\le
    \gamma\norm{\Pi_x\ket{\Gamma_x}}
    +
    2\norm{(I-\Pi_x)\ket{\Gamma_x}}\\
    &\le
    \gamma + 2\sqrt{\delta}\\
    &=
    \frac{\eps}{2\sqrt{2^n}},
    \end{align*}
    by the triangle inequality, \Cref{eq:opnorm_proj}, \Cref{eq:proj_diff_delta}, and \Cref{eq:def_gamma_delta}, respectively.
    
    Applying $A^\dagger_{n, \eps}$ preserves Euclidean distance, so
    \[
    \norm{
    C_{n,\eps}\ket{x}\ket{0^a}
    -
    R_{f,n}\ket{x}\ket{0^a}
    }
    \le
    \frac{\eps}{2\sqrt{2^n}}.
    \]
    Appealing to \Cref{lem:basis_to_op_norm} and \Cref{lem:extension_close_op_norm} completes the proof, as $C_{n,\eps}$ is clearly space-uniform.
\end{proof}

The previous lemma is particularly useful in the context of unitary complexity because when modeling $f\Phase$ as a perfect unitary oracle, we will not have to keep track of the approximation error $\eps$, thus simplifying subsequent proofs.

Observe that certain operations on encoding functions become very easy in the phase encoding.
For example, addition of two functions $f + g$ corresponds to the unitary multiplication $f\Phase \cdot g\Phase$ (as long as the sum remains bounded in $[-\pi/2, \pi/2]$).
But one can perform more general transformations as well, including composition with any function of bounded \textit{modulus of continuity}:

\begin{definition}[Modulus of continuity]
    \label{def:modulus_continuity}
    A function $g: \R \to \R$ has \textit{polylogarithmic modulus of continuity} if there is a polynomial $p$ such that for every $x, y$ and $0 < \eps < 1$, $\abs{x - y} \le 2^{-p(\log 1/\eps)}$ implies $\abs{g(x) - g(y)} \le \eps$.
\end{definition}

As a special case, note that all Lipschitz-continuous functions have polylogarithmic modulus of continuity.

\begin{lemma}
    \label{lem:continuity_composition}
    Suppose $f : \{0,1\}^* \to [-\pi/2, \pi/2]$ satisfies $f\Phase \in \uPSPACE^\calO$.
    Let $g : [-\pi/2, \pi/2] \to [-\pi/2, \pi/2]$ be a function with polylogarithmic modulus of continuity that also has an FPTEAS on inputs encoded in fixed-point binary.
    Then $(g \circ f)\Phase \in \uPSPACE^\calO$.
\end{lemma}

\begin{proof}
    By \Cref{lem:rotation_to_fqpseas}, it suffices to show that $g \circ f$ has an FPUSEAS relative to $\calO$.
    Let $p$ be the polynomial from \Cref{def:modulus_continuity} that controls the continuity of $g$.
    Given a desired error parameter $\eps$, begin by running the FPUSEAS for $f$ that is implied by \Cref{lem:rotation_to_fqpseas} with error parameter $\gamma = 2^{-p(\log(2/\eps))}$.
    This produces an output $\phi$ that with probability at least $2/3$ satisfies $\abs{\phi - f(x)} \le \gamma$, and in turn that $\abs{g(\phi) - g(f(x))} \le \eps / 2$.
    Finally, output the result of the FPTEAS $M(\phi, \eps/2)$ for $g$ which satisfies $\abs{M(\phi, \eps/2) - g(\phi)} \le \eps/2$.
    By the triangle inequality, the output is within $\eps$ of $g(f(x))$ with probability at least $2/3$, giving an FPUSEAS for $g \circ f$.
\end{proof}

\subsection{Entrywise characterization}

We now turn to characterizing $\uPSPACE$ in terms of the computability of its entries.
To do so, we use phase encodings of the entries.

\begin{definition}[$\calU\Entry$]
    \label{def:u_entry}
    Let $\calU = \{U_n\}_{n \in \N}$ be a unitary language.
    Then $\calU\Entry: \{0,1\}^{2n + 1} \to [-1,1]$ is the encoding function
    \begin{align*}
    \calU\Entry(x, y, 0) &\coloneqq \Re\;\braket{y|U_n}{x},\\
    \calU\Entry(x, y, 1) &\coloneqq \Im\;\braket{y|U_n}{x}.
    \end{align*}
\end{definition}

Standard amplitude estimation techniques allow us to space-efficiently implement $\calU\Entry\Phase$ with oracle access to $\calU$.

\begin{lemma}[Entries computable in PSPACE]
    \label{lem:u_to_entry}
    $\calU\Entry\Phase \in \uPSPACE^{\calU}$.
\end{lemma}

\begin{proof}
    Since $\calU\Entry$ takes values in $[-1,1]$, by \Cref{lem:rotation_to_fqpseas} it suffices to show that $\calU\Entry$ has an FPUSEAS relative to $\calU$.

    Fix $x,y \in \{0,1\}^n$, $b \in \{0,1\}$, and error parameter $\eps>0$.
    First consider the following Hadamard-test circuit:
    \[
    \Qcircuit @C=1em @R=1em {
    \lstick{\ket{x}} &  \qw & \ctrl{4} & \qw & \qw & \qw & \qw & \qw\\
    \lstick{\ket{y}} &  \qw & \qw & \ctrl{2} & \qw & \qw & \qw & \qw\\
    \lstick{\ket{b}} & \qw & \qw & \qw & \qw & \ctrl{1} & \qw & \qw\\
    \lstick{\ket{0}} & \gate{H} & \ctrl{1} & \ctrlo{1} & \ctrl{1} & \gate{S^\dagger} & \gate{H} & \meter\\
    \lstick{\ket{0^n}} &  \qw & \targ & \targ & \gate{U_n} & \qw & \qw & \qw\\
    }
    \]
    To be clear, the two doubly-controlled gates each represent $n$ parallel gates on the corresponding wires of the $n$-qubit registers.
    Ignoring the unchanged input registers, the state immediately before measurement is
    \[
    \frac{\ket{+}\ket{y} + (i^{-b})\ket{-}U_n\ket{x}}{\sqrt{2}}.
    \]
    Therefore, if $b=0$, the probability of obtaining outcome $\ket{0}$ is
    \[
    p_0
    \coloneqq
    \frac{1}{4}\norm{\ket{y}+U_n\ket{x}}^2
    =
    \frac{1+\Re\bra{y}U_n\ket{x}}{2}.
    \]
    If $b=1$, then the $S^\dagger$ gate gives
    \[
    p_1
    \coloneqq
    \frac{1}{4}\norm{\ket{y}-iU_n\ket{x}}^2
    =
    \frac{1+\Im\bra{y}U_n\ket{x}}{2}.
    \]
    Thus, in both cases, $\calU\Entry(x,y,b)=2p_b-1$.

    We now estimate $p_b$ using amplitude estimation~\cite{BHMT02-amp}.
    Choose the parameters so that it returns a number $\widetilde p_b$ satisfying $\abs{\widetilde p_b-p_b}\le \eps/2$ with probability at least $2/3$, then output $\phi \coloneqq 2\widetilde p_b-1$.
    With probability at least $2/3$,
    \[
    \abs{\phi-\calU\Entry(x,y,b)}
    =
    2\abs{\widetilde p_b-p_b}
    \le \eps.
    \]
    The amplitude-estimation circuit uses $O(\log 1/\eps)$ additional qubits beyond the $3n + 2$ required for the circuit above, so the overall space complexity remains $\poly(n,\log 1/\eps)$.
    Therefore, $\calU\Entry$ has an FPUSEAS relative to $\calU$.
\end{proof}

Next, we will show that the ability to compute entries of $\calU$ allows us to prepare any \textit{column} of $\calU$ with polynomial space overhead.
We will express this result in terms of a canonical transformation that encodes the columns of a given $\calU$:

\begin{definition}[$\calU\Column$]
    \label{def:column}
    Let $\calU = \{U_n\}_{n \in \N}$ be a unitary language.
    Denote by
    \[
    \ket{\psi_x} \coloneqq \frac{\ket{0}\ket{0^n} - \ket{1}U_n\ket{x}}{\sqrt{2}}.
    \]
    Then $\calU\Column = \{V_{2n + 1}\}_{n \in \N}$ is the unitary language defined by
    \[
    V_{2n+1} = \sum_{x \in \{0,1\}^n} \ketbra{x}{x} \otimes (I - 2\ketbra{\psi_x}{\psi_x}).
    \]
\end{definition}

\begin{remark}
$V_{2n+1}$ allows one to prepare any column of $U_n$ in the sense that it maps $\ket{x}\ket{0}\ket{0^n}$ to $\ket{x}\ket{1}U_n\ket{x}$ for every $x \in \{0,1\}^n$.
Note, however, that the initial $\ket{x}$ remains as garbage that would need to be erased, and so $\calU\Column$ does not immediately allow for implementing $\calU$.
\end{remark}

Our next proof will use the following result as a black box:

\begin{lemma}[{Exact amplitude amplification~\cite{MGB26-preparation}}]
    \label{lem:exact_amp}
    Let $U$ be an $(n+1)$-qubit unitary which satisfies
    \[
    (I \otimes \ketbra{0}{0}) U\ket{0^n}\ket{0} = a\ket{\psi}\ket{0}
    \]
    for some $n$-qubit $\ket{\psi}$ and $a > 0$.
    Define
    \[
    k \coloneqq \left\lceil \frac{\pi}{4\arcsin a} - \frac{1}{2}\right\rceil, \qquad \theta \coloneqq \frac{\pi}{4k + 2},
    \]
    and let $R$ be any single-qubit unitary for which $\braket{0|R}{0} = \frac{\sin \theta}{a}$.
    Then the circuit below outputs $(-1)^k\ket{\psi}\ket{0}\ket{0}$:
    \[
    \Qcircuit @C=1em @R=1em {
     \lstick{\ket{0^n}} & \multigate{1}{U} & \qw &\qw & \multigate{1}{U^\dagger} & \ctrlo{1} & \multigate{1}{U} & \qw & \qw\\
     \lstick{\ket{0}} & \ghost{U} & \qw & \ctrlo{1} & \ghost{U^\dagger} & \ctrlo{1} & \ghost{U} & \qw & \qw\\
     \lstick{\ket{0}} & \gate{R} & \qw & \ctrlo{-1} & \gate{R^\dagger} & \ctrlo{-1} & \gate{R} & \qw & \qw
     \gategroup{1}{4}{3}{7}{2em}{.}\\
     &&&&& \raisebox{-2em}{\normalfont \hspace{-2em} Repeat sequentially $k$ times}\\
     \vspace{2em}
    }
    \]
\end{lemma}

Below, we show that oracle access to $\calU\Entry\Phase$ allows us to prepare the columns of $\calU$ in polynomial space.

\begin{lemma}[Columns computable from entries]
    \label{lem:entry_to_column}
    $\calU\Column \in \uPSPACE^{\calU\Entry\Phase}$.
\end{lemma}

In some sense, this lemma is not new; it is already known that states can be prepared in polynomial space with access to the entries~\cite{GR02-state-prep,Aar16-barbados}.
However, these algorithms generally require access to the \textit{ratios} of sums involving entries, which cannot be computed stably from approximate entries.
This does not substantially affect these algorithms' overall correctness, because the numerically unstable computations involve negligible fractions of the superposition.
Still, accounting for the error in each step is somewhat cumbersome.
We take a different approach based on amplitude amplification, which has the benefit that we do not have to directly account for error anywhere in the proof.
Indeed, the proof that follows illustrates how simple our definitions make it to compose unitary algorithms.

\begin{proof}[Proof of \Cref{lem:entry_to_column}]
    We claim that it suffices to implement any unitary $C_n$ which satisfies
    \begin{equation}
    \label{eq:qram}
    C_n\ket{x}\ket{0^n}\ket{0^a} = \ket{x}U_n\ket{x}\ket{0^a}
    \end{equation}
    on every $x \in \{0,1\}^n$.
    This is because the circuit
    \begin{equation}
    \label{eq:qram_canonical_ckt}
    \Qcircuit @C=1em @R=1em {
     & \multigate{2}{C_n^\dagger} & \qw & \multigate{2}{C_n} & \qw\\
     & \ghost{C_n^\dagger} & \ctrlo{1} & \ghost{C_n} & \qw\\
     & \ghost{C_n^\dagger} & \ctrlo{1} & \ghost{C_n} & \qw\\
     & \ctrl{-1} & \targ & \ctrl{-1} & \qw
    }
    \end{equation}
    is an $a$-qubit extension of $V_{2n + 1}$ (with the third register acting as the ancilla).
    To build such a $C_n$, our strategy will be to first design a circuit that uses postselection to prepare $U_n\ket{x}$ with an inverse-exponential success probability.
    Then, we will use amplitude amplification to increase the success probability to $1$, which will crucially rely on having a success probability that is independent of $U_n$ or $\ket{x}$.

    Let $g \coloneqq \arcsin \circ\; \calU\Entry$, and consider the circuit below:
    \[
    \Qcircuit @C=1em @R=1em {
     \lstick{\ket{x}} & \qw & \qw & \multigate{2}{g\Phase} & \multigate{2}{g\Phase^\dagger} & \qw & \qw & \qw\\
     \lstick{\ket{0^n}} & \gate{H^{\otimes n}} & \qw &  \ghost{g\Phase} & \ghost{g\Phase^\dagger} & \qw & \qw & \qw\\
     \lstick{\ket{0}} & \gate{H} & \qw & \ghost{g\Phase} & \ghost{g\Phase^\dagger} & \qw & \qw & \qw\\
     \lstick{\ket{0}} & \gate{S^\dagger} & \gate{H} & \ctrl{-1} & \ctrlo{-1} & \gate{H} & \gate{S} & \qw\\
    }
    \]
    The rough idea is that the circuit applies a rotation on the $\ket{0}/\ket{1}$ plane of the last qubit conditioned on the value that $\calU\Entry$ takes on the first three registers.
    This exploits the circuit identity:
    \[
    \begin{bmatrix}
        \cos \theta & -\sin \theta\\
        \sin \theta & \cos \theta
    \end{bmatrix}
    =
    SH \begin{bmatrix}
        e^{-i \theta} & 0\\
        0 & e^{i \theta}
    \end{bmatrix}
    HS^\dagger.
    \]
    Consequently, if we define $a_{xy}$ and $b_{xy}$ to be the real and imaginary parts of $\braket{y}{U_n|x}$, the output state of the above circuit is exactly:
    \[
    \ket{x} \frac{1}{\sqrt{2^{n + 1}}}\sum_{y \in \{0,1\}^n} \ket{y} \left(\sqrt{1 - a_{xy}^2}\ket{00} + a_{xy}\ket{01} + \sqrt{1 - b_{xy}^2}\ket{10} + b_{xy}\ket{11}\right).
    \]
    Now consider what happens when we append a few gates and measure:
    \[
    \Qcircuit @C=1em @R=1em {
     \lstick{\ket{x}} & \qw & \qw & \multigate{2}{g\Phase} & \multigate{2}{g\Phase^\dagger} & \qw & \qw & \qw\\
     \lstick{\ket{0^n}} & \gate{H^{\otimes n}} & \qw &  \ghost{g\Phase} & \ghost{g\Phase^\dagger} & \qw & \qw & \qw\\
     \lstick{\ket{0}} & \gate{H} & \qw & \ghost{g\Phase} & \ghost{g\Phase^\dagger} & \gate{S} & \gate{H} & \qw & \meter\\
     \lstick{\ket{0}} & \gate{S^\dagger} & \gate{H} & \ctrl{-1} & \ctrlo{-1} & \gate{H} & \gate{S} & \gate{X} & \meter\\
    }
    \]
    Conditioned on obtaining $\ket{0}$ from both measurements, the unnormalized state on the remaining registers is exactly the state that we want $C_n$ to output:
    \[
    \ket{x} \frac{1}{\sqrt{2^{n + 2}}}\sum_{y \in \{0,1\}^n} (a_{xy} + i b_{xy})\ket{y} = \frac{\ket{x}U_n\ket{x}}{\sqrt{2^{n  + 2}}}.
    \]
    Therefore, we can use exact amplitude amplification (\Cref{lem:exact_amp}) to obtain a space-uniform $g\Phase$-oracle unitary circuit $C_n$ that behaves as desired (\Cref{eq:qram}), as $k = O(\sqrt{2^n})$.
    Note that there is no error dependence: the resulting circuit (\Cref{eq:qram_canonical_ckt}) extends $V_{2n+1}$ exactly, and can be the same circuit for all $\eps$.
    This establishes that $\calU\Column \in \uPSPACE^{g\Phase}$.

    To complete the proof, we note that \Cref{lem:continuity_composition} puts $g\Phase \in \uPSPACE^{\calU\Entry\Phase}$, because the $\arcsin$ function satisfies $\abs{\arcsin x - \arcsin y} \le \pi \sqrt{\abs{x - y}}$.
    Finally, because of \Cref{prop:tidy_low}, we have
    \[
    \calU\Column \in \uPSPACE^{g\Phase} \subseteq \uPSPACE^{\uPSPACE^{\calU\Entry\Phase}} = \uPSPACE^{\calU\Entry\Phase}.\qedhere
    \]
\end{proof}

Finally, we note that Rosenthal~\cite{Ros21-unitary} showed how to space-efficiently simulate a unitary with access to an oracle that prepares its columns.
In particular:

\begin{lemma}[Unitaries computable from columns]
    \label{lem:column_to_u}
    $\calU \in \uPSPACE^{\calU\Column}$.
\end{lemma}

\begin{proof}
    This is exactly \cite[Theorem 2.2]{Ros21-unitary}, so we will not reproduce the proof entirely.
    The strategy can be summarized roughly as follows: first, use the $\calU\Column$ oracle to map basis states $\ket{x}\ket{0^n}$ to $\ket{x}U_n\ket{x}$.
    Then, use $\calU\Column$ as a reflection oracle to erase the $\ket{x}$ on the first register conditioned on the second register being $U_n\ket{x}$.
    The erasure step uses exact amplitude amplification (\Cref{lem:exact_amp}) and is therefore errorless.
    One can also easily verify that Rosenthal's circuit satisfies \Cref{def:unitaryP_tidy_oracle}.
\end{proof}

Combining the above lemmas shows that there are many equivalent characterizations of $\uPSPACE$ in terms of unitary entries:

\begin{theorem}[Characterization of unitary polynomial space]
    \label{thm:entrywise_characterization}
    The following are equivalent for any unitary language $\calU$ and unitary oracle $\calO$:
    \begin{enumerate}
        \item $\calU \in \uPSPACE^{\calO}$.\label{item:U}
        \item $\calU\Entry\Phase \in \uPSPACE^{\calO}$.\label{item:UEntry}
        \item $\calU\Column \in \uPSPACE^{\calO}$.\label{item:UColumn}
        \item $\calU\Entry$ has an FPUSEAS relative to $\calO$.\label{item:FPUSEAS}
        \item $\calU\Entry\Apx \in \BQuPSPACE^\calO$.\label{item:Apx}
    \end{enumerate}
    If $\calO$ is additionally a (classical) language $L$, then the above are equivalent to:
    \begin{enumerate}
        \setcounter{enumi}{5}
        \item $\calU\Entry$ has an FPSEAS relative to $L$.\label{item:FPSEAS}
        \item $\calU\Entry\Apx \in \PSPACE^L$.\label{item:PSPACE}
    \end{enumerate}
\end{theorem}

\begin{proof}
    In light of \Cref{prop:tidy_low}, the equivalence between \Cref{item:U,item:UEntry,item:UColumn} follows from   \Cref{lem:u_to_entry,lem:entry_to_column,lem:column_to_u}.
    \Cref{lem:rotation_to_fqpseas} implies the equivalence between \Cref{item:UEntry,item:FPUSEAS}, and \Cref{prop:fapx_iff_fp} implies the equivalence between \Cref{item:FPUSEAS,item:Apx}.
    The equivalence between \Cref{item:FPUSEAS,item:FPSEAS} is exactly \Cref{cor:fpuseas_in_pspace}.
    Lastly, \Cref{prop:bqpspace} establishes an equivalence between \Cref{item:Apx,item:PSPACE}.
\end{proof}

\begin{remark}
    \label{rmk:pspace_purification}
    We can now see why the requirement of a unitary implementation in the definition of $\puPSPACE$ (\Cref{def:unitaryP_garbage}) is without loss of generality.
    Consider the following chain of implications:
    \begin{align*}
        \calU \in \puPSPACE
        &\iff \calU \otimes \calU^* \in \uPSPACE && (\mathrm{\Cref{thm:u_garbage_iff_u_ustar_tidy}})\\
        &\iff (\calU \otimes \calU^*)\Entry\Apx \in \BQuPSPACE && (\mathrm{\Cref{thm:entrywise_characterization}})\\
        &\iff (\calU \otimes \calU^*)\Entry\Apx \in \BQPSPACE && (\mathrm{\Cref{prop:bqpspace}}).
    \end{align*}
    Note that $U_n \otimes U_n^*$ is the matrix representation on the operator basis for the channel $\Lambda_n(\rho) = U_n\rho U_n^\dagger$, and any algorithm implementing $\Lambda_n$ to arbitrarily small error can be used to approximate the entries of its matrix representation with polynomial space overhead.
    Therefore, alternative definitions of $\puPSPACE$ that allow implementation by general nonunitary circuits would coincide with \Cref{def:unitaryP_garbage}.
\end{remark}

\subsection{Erasing garbage}

We conclude this section by using \Cref{thm:entrywise_characterization} to prove a (near) equivalence between $\uPSPACE$ and $\puPSPACE$.
As we saw before, the example
\[
\calU \coloneqq \{(-1)^{HALT(n)}I^{\otimes n}\}_{n \in \N},
\]
establishes that $\uPSPACE \neq \puPSPACE$: there is simply no way for $\uPSPACE$ to ``know'' the right global phase.
The next theorem establishes that issues of global phase are the \textit{only} obstruction, in the sense that every unitary in $\puPSPACE$ is contained in $\uPSPACE$ up to a global phase.
Consequently, we show that one can generically erase the garbage produced when implementing a unitary in polynomial \textit{space}, even though we do not know whether the same is possible in polynomial \textit{time}.
In our view, this justifies exclusively using the tidy definition of unitary polynomial space in any subsequent investigation on the subject, because it is without loss of generality to pick a garbage-free implementation.

\begin{theorem}
    \label{thm:projunitarypspace_equiv}
    For every language $L$, $\calU \in \puPSPACE^L$ if and only if there exists a $\calV \in \uPSPACE^L$ and a sequence $\{\phi_n \in \C\}_{n \in \N}$ such that $U_n = \phi_n V_n$.
\end{theorem}

\begin{proof}
    Suppose $\calU \in \puPSPACE^L$.
    By \Cref{thm:u_garbage_iff_u_udag_tidy_oracle} we have
    \[
    \calU \otimes \calU^\dagger \in \uPSPACE^{L \otimes L} \subseteq \uPSPACE^L.
    \]
    Thus, by the preceding \Cref{thm:entrywise_characterization}, $(\calU \otimes \calU^\dagger)\Entry$ has an FPSEAS relative to $L$.

    We now choose a canonical phase for each $U_n$.
    Let $N=2^n$.
    For $i,j\in\{0,1\}^n$, write
    \[
    q_{ij} \coloneqq \abs{\bra{j}U_n\ket{i}}^2.
    \]
    Notice that $q_{ij}$ is an entry of $U_n \otimes U_n^\dagger$.
    Using the FPSEAS for entries of $U_n \otimes U_n^\dagger$, approximate each $q_{ij}$ to additive error at most $1/(3N)$, and let $(i_n,j_n)$ be the lexicographically first pair for which the approximation is at least $2/(3N)$.
    Such a pair exists because the columns of $U_n$ are normalized to $1$, so there exists $(i,j)$ with $q_{ij} \ge 1/N$.
    For the selected pair $(i_n,j_n)$ we have
    \[
    q_{i_nj_n} \ge \frac{1}{3N}.
    \]
    Define
    \[
    \phi_n \coloneqq \frac{\bra{j_n}U_n\ket{i_n}}{\abs{\bra{j_n}U_n\ket{i_n}}},
    \qquad
    V_n \coloneqq \phi_n^* U_n.
    \]
    Then $V_n$ is unitary and $U_n = \phi_n V_n$.
    It remains to show that $\calV=\{V_n\}_{n\in\N}$ lies in $\uPSPACE^L$.

    By \Cref{thm:entrywise_characterization}, it suffices to show that $\calV\Entry$ has an FPSEAS relative to $L$.
    Fix $x,y\in\{0,1\}^n$.
    We have
    \[
    \bra{y}V_n\ket{x}
    =
    \phi_n^*\bra{y}U_n\ket{x}
    =
    \frac{\bra{i_n}U_n^\dagger \ket{j_n}\bra{y}U_n\ket{x}}{\abs{\bra{j_n} U_n \ket{i_n}}}.
    \]
    Notice that the numerator is an entry of $U_n \otimes U_n^\dagger$, and the denominator is the square root of the entry $q_{i_nj_n}$.
    Given target error parameter $\eps > 0$, let
    \[
    \eta \coloneqq \frac{\eps}{2\sqrt{3N}}, \qquad \delta \coloneqq \frac{\eps}{2\sqrt{2}(3N)^{3/2}}.
    \]
    Use the FPSEAS to estimate the numerator to error $\eta$ and $q_{i_nj_n}$ to error $\delta$.
    Assume without loss of generality that we truncate the estimate of the numerator to have absolute value at most $1$.
    Then the following claim establishes that the overall estimate is within $\eps$ of $\bra{y}V_n\ket{x}$:

    \begin{claim}
        Suppose $b > 0$, $\abs{a - a'} \le \eta$, and $\abs{b - b'} \le \delta \le b / 2$.
        Then
        \[
        \abs[\Big]{\frac{a}{\sqrt{b}} - \frac{a'}{\sqrt{b'}}} \le \frac{\eta}{\sqrt{b}} + \frac{\sqrt{2} \delta}{b^{3/2}}\abs{a'}.
        \]
    \end{claim}

    \begin{proof}[Proof of claim]
        \begin{align*}
            \abs[\Big]{\frac{a}{\sqrt{b}} - \frac{a'}{\sqrt{b'}}}
            &= \left\vert\left(\frac{a}{\sqrt{b}} - \frac{a'}{\sqrt{b}} \right) + \left(\frac{a'}{\sqrt{b}} - \frac{a'}{\sqrt{b'}} \right)\right\vert\\
            &\le \abs[\Big]{\frac{a - a'}{\sqrt{b}}} + \abs{a'}\abs[\Big]{\frac{1}{\sqrt{b}} - \frac{1}{\sqrt{b'}}}\\
            &\le \frac{\eta}{\sqrt{b}} + \frac{\sqrt{2} \delta}{b^{3/2}}\abs{a'}
        \end{align*}
        The last line follows from the mean value theorem: we know that $b' \ge b/2 > 0$, and therefore the function $t \to t^{-1/2}$ has derivative bounded by
        \[
        \frac{1}{2}(b/2)^{-3/2}
        =
        \frac{\sqrt{2}}{b^{3/2}}
        \]
        on the interval between $b$ and $b'$.
    \end{proof}

    The index pair $(i_n,j_n)$ itself can be found in space $\poly(n)$ by the lexicographic search above, and running the FPSEAS for $(\calU\otimes\calU^\dagger)\Entry$ and performing the resulting arithmetic takes space $\poly(n, \log 1/\eps)$.
    Thus we have exhibited an FPSEAS for $\calV\Entry$, and applying \Cref{thm:entrywise_characterization} once more gives $\calV \in \uPSPACE^L$.

    The converse is immediate: the $\uPSPACE^L$ implementation yields a $\puPSPACE^L$ implementation in which we incorporate the phase $\phi_n$ as a $0$-qubit garbage state.
    And even though $\uPSPACE^L$ allows for controlled queries to $L$ and $\puPSPACE^L$ does not, we can simulate controlled queries to $L$ with projective queries:
    \[
    \Qcircuit @C=1em @R=1em {
    \lstick{\ket{c}} & \qw & \ctrl{3} & \qw & \qw & \ket{c}\\
    \lstick{\ket{x}} & \multigate{1}{L} & \qw & \multigate{1}{L} & \qw & \ket{x}\\
    \lstick{\ket{0}} & \ghost{L} & \ctrl{1} & \ghost{L} &  \qw & \ket{0}\\
    \lstick{\ket{b}} & \qw & \targ & \qw & \qw & \rstick{\hspace{-1.1em}\ket{b \oplus c \cdot L(x)}}
    }
    \]
    Ignoring the ancilla qubit that resets to $\ket{0}$, the circuit above applies controlled-$L$ on the first, second, and fourth registers.
\end{proof}

One might wonder why \Cref{thm:projunitarypspace_equiv} requires the oracle to be a classical language rather than a more general unitary oracle.
In fact, we show that our theorem provably \textit{fails} with unitary oracles.
The proof is closely related to work by Gavorov\'{a}, Seidel, and Touati~\cite{gst24}, which uses topology to show that one cannot generically lift uncontrolled queries of a unitary into controlled queries.
Their key lemma (stated below) shows that there is no continuous way to assign a global phase to every projective unitary.

\begin{lemma}[{\cite[Lemma 3]{gst24}}]
    \label{lem:no_global_phase}
    Fix $N \ge 2$ and let $U(N)$ be the $N$-dimensional unitary group, and $PU(N)$ be the $N$-dimensional projective unitary group.
    Let $\pi : U(N) \to PU(N)$ be the quotient map.
    There is no continuous map $\sigma: PU(N) \to U(N)$ such that $\pi(\sigma(g)) = g$ for every $g \in PU(N)$.
\end{lemma}

More generally, we will need the non-existence of any continuous map that \textit{approximately} assigns a global phase to every projective unitary.

\begin{lemma}
\label{lem:no_approximate_projective_section}
Fix $N \ge 2$ and let $\pi : U(N) \to PU(N)$ be the quotient map.
Let $0 \le \eps < 1$.
There is no continuous map $A : PU(N) \to \C^{N \times N}$ such that, for every $g \in PU(N)$, there exists a representative $U \in U(N)$ of $g$ and a phase $\phi \in \C$ with $\abs{\phi} = 1$ for which
\[
\opnorm{A(g) - \phi U} \le \eps.
\]
\end{lemma}

\begin{proof}
    Towards a contradiction, suppose such an $A$ exists.
    Let $B: PU(N) \to U(N)$ be any map which chooses a unitary representative, meaning $\pi(B(g)) = g$ for every $g$.
    (Note: \Cref{lem:no_global_phase} shows that $B$ cannot be \textit{continuous}, but such a $B$ certainly exists; for example, by choosing the phase so that the first nonzero entry is real and positive.)

    Define $z: PU(N) \to \C$ by
    \[
    z(g) \coloneqq \tr(B(g)^\dagger A(g)).
    \]
    We first prove that $z(g)$ is never zero.
    By assumption, there exists a phase $\phi$ such that
    \[
        \opnorm{A(g)-\phi B(g)} \le \eps.
    \]
    Hence
    \[
        \abs{\tr(B(g)^\dagger A(g)) - \phi N}
        \le
        N \opnorm{B(g)^\dagger A(g)-\phi I}
        \le
        N\eps,
    \]
    so
    \[
        \abs{z(g)} \ge N(1-\eps) > 0.
    \]

    Therefore, the following function $\sigma: PU(N) \to U(N)$ is well-defined:
    \[
        \sigma(g) \coloneqq \frac{z(g)}{\abs{z(g)}} B(g).
    \]
    First observe that $\sigma$ is independent of the choice of $B$: replacing $B(g)$ by $\omega B(g)$ replaces $z(g)$ by $\overline{\omega}z(g)$, and therefore leaves $\sigma(g)$ unchanged.
    Consequently, $\sigma$ is continuous, as can be checked on any local choice of representative $B$.
    Finally, notice that $\pi(\sigma(g))=g$ for every $g \in PU(N)$, which contradicts \Cref{lem:no_global_phase}.
\end{proof}

Our last theorem in this section follows by combining \Cref{lem:no_approximate_projective_section} with a standard diagonalization.

\begin{theorem}
    \label{thm:no_global_phase}
    There is a unitary language $\calU$ such that for every unitary language $\calV = \{\phi_n U_n\}_{n \in \N}$ that equals $\calU$ up to a sequence of global phases $\{\phi_n\}_{n \in \N}$, we have $\calV \not\in \uPSPACE^{\calU \otimes \calU^\dagger}$.
\end{theorem}

\begin{remark}
    Note that $\calV \in \puPSPACE^{\calU \otimes \calU^\dagger}$, as $\puPSPACE$ can leave the garbage generated on the $\calU^\dagger$ register.
    Thus, \Cref{thm:no_global_phase} shows that \Cref{thm:projunitarypspace_equiv} cannot relativize in the presence of unitary oracles.
\end{remark}

\begin{proof}[Proof of \Cref{thm:no_global_phase}]
    We construct $\calU=\{U_n\}_{n\in\N}$ in stages by diagonalizing against all space-uniform $(\calU\otimes\calU^\dagger)$-oracle unitary circuit families.
    During the construction, some lengths $n \in \N$ will be declared ``frozen'', meaning that their corresponding unitaries have already been fixed and will never again be changed.
    Initially, no lengths are frozen.
    
    At stage $t$, denote the $t$th space-uniform $(\calU\otimes\calU^\dagger)$-oracle unitary circuit family by $\{C^{(t)}_{n,\eps}\}_{n \in \N, \eps > 0}$.
    Choose an input length $n$ larger than every length frozen in the previous stages, and consider the circuit $C^{(t)}_{n,\eps}$, where $\eps$ is any fixed constant (say, $1/2$).
    All oracle lengths $m$ other than $n$ that appear in this circuit are now fixed as follows: if the corresponding unitary $U_m$ has already been defined, keep its old value; otherwise set it equal to the identity.
    We will freeze all of these lengths $m$ at the end of the stage.

    Now vary only the still-undefined unitary $U_{n}$.
    Since the oracle available to the circuit is $U_{n}\otimes U_{n}^\dagger$, its action depends only on the projective class of $U_{n}$.
    After fixing all other queried oracle lengths, the $2^n \times 2^n$ zero-ancilla block of $C^{(t)}_{n,\eps}$ defines a map
    \[
        A : PU(2^{n}) \to \C^{2^{n} \times 2^{n}}
    \]
    by instantiating the circuit with an oracle $U_{n}\otimes U_{n}^\dagger$ that agrees with the input projective unitary $g \in PU(2^n)$.
    $A$ is continuous because the circuit contains only finitely many oracle gates, each of which depends continuously on $U_n\otimes U_n^\dagger$ and hence on $g \in PU(2^n)$.

    By \Cref{lem:no_approximate_projective_section}, there exists some $g \in PU(2^n)$ with corresponding unitary $U_{n}\in U(2^n)$ such that
    \[
        \min_{\abs{\phi}=1}
        \opnorm{A(g)-\phi U_{n}}
        >
        \eps .
    \]
    Define $U_{n}$ to be this unitary.
    Finally, freeze $n$ and every oracle length queried by
    $C^{(t)}_{n,\eps}$.
    This completes stage $t$.
    After all stages have been completed, define $U_n=I$ for every length $n$ that was never assigned a value during the construction.

    We can now see that $\calU$ has the desired property.
    Suppose, toward a contradiction, that there is a sequence of phases
    $\{\phi_n\}_{n\in\N}$ such that
    \[
        \calV=\{\phi_n U_n\}_{n\in\N}
        \in
        \uPSPACE^{\calU\otimes\calU^\dagger}.
    \]
    Then some space-uniform oracle circuit family $\{C^{(t)}_{n, \eps}\}_{n \in \N, \eps > 0}$ in our enumeration witnesses this containment.
    Since the circuit family is assumed to implement $\calV$, the circuit $C^{(t)}_{n,\eps}$ must be $\eps$-close in operator norm to an extension $\phi_n U_n \oplus W$ of $\phi_n U_n$.
    Taking the zero-ancilla block of this extension gives exactly $\phi_{n} U_{n}$, so the zero-ancilla block $A$ of
    $C^{(t)}_{n,\eps}$ must satisfy
    \[
        \opnorm{A - \phi_{n}U_{n}} \le \eps.
    \]
    But this contradicts the way $U_{n}$ was chosen at stage $t$.
    Therefore no such $\calV$ lies in
    $\uPSPACE^{\calU\otimes\calU^\dagger}$.
\end{proof}

\begin{remark}
\label{rmk:topology}
The proof of \Cref{thm:no_global_phase} is significantly more general than its statement, as the proof uses only two properties of $\uPSPACE$: one, that the algorithm evaluates the target unitary to some nontrivial precision, and two, that the algorithm's behavior is a continuous function of $\calU \otimes \calU^\dagger$.
\end{remark}

\section{Unitary complexity and the unitary synthesis problem}
\label{sec:synthesis}

In this section, we consider the relationship between $\uP$ and $\uPSPACE$ vis-a-vis different possible resolutions to the unitary synthesis problem.
We formulate unitary synthesis as follows:

\begin{problem}[Unitary synthesis problem~\cite{AK07-qcma-qma}]
    Does $\uALL = \uP^\ALL$?
\end{problem}

In other words, the unitary synthesis problem asks whether every unitary language can be implemented efficiently relative to \textit{some} classical oracle.

\subsection{\texorpdfstring{$\NEXP$}{NEXP} relation problems}
\label{sec:nexp_relations}

We will establish a close relation between the oracles needed to efficiently synthesize $\uPSPACE$ and $\NEXP$ witnesses.
To express our results succinctly, we first fix some notation for $\NEXP$ relation problems.

\begin{definition}
    A relation $R \subseteq \{0,1\}^* \times \{0,1\}^*$ is in $\FNEXP^{L}$ if there is an $L$-oracle Turing machine $M(x, y)$ that runs in time $2^{\poly(|x|)}$ and whose output indicates whether $(x, y) \in R$, and $|y| \le 2^{\poly(|x|)}$ for every $(x, y) \in R$.
    $R$ is in $\FNEXP^{L[\poly]}$ if there exists such an $M$ that is additionally restricted to only query $L$ on strings of length at most $\poly(|x|)$.
\end{definition}

\begin{definition}
    A language $W$ is a \emph{witness oracle} for a relation $R$ if, for every $x$ for which $\{(x, \cdot) \in R\} \neq \emptyset$, the string $y$ defined by
    \[
    y_i \coloneqq W(x,i)
    \]
    satisfies $(x, y) \in R$.
    On inputs $x$ for which no valid $y$ exists, $W(x,i)$ is unconstrained.
\end{definition}

In a slight abuse of notation, when $\class{K}$ is a complexity class that can be augmented with an oracle and $R$ is a relation, we use $\class{K}^R$ to mean $\class{K}$ with access to an arbitrary witness oracle $W$ for $R$, with a fixed $\class{K}$ base machine that is correct for any choice of $W$.
When $\class{R}$ is a class of relations, $\class{K}^\class{R}$ is the union thereof, i.e.
\[
\class{K}^\class{R} \coloneqq \bigcup_{R \in \class{R}} \class{K}^R.
\]

\subsection{Bounds on unitary synthesis oracles}

This section's key lemma proves a dichotomy: every unitary language in $\uPSPACE$ either cannot be synthesized efficiently relative to \textit{any} oracle, or else it can, in which case an $\FNEXP$ witness oracle suffices.

\begin{lemma}
    \label{lem:unitarypspace_dichotomy_witnexp}
    For every language $L$ and every $\calU \in \uPSPACE^L$, either
    \begin{enumerate}
        \item $\calU \notin \uP^{\ALL}$, or
        \item $\calU \in \uP^{\FNEXP^{L[\poly]}}$.
    \end{enumerate}
\end{lemma}

\begin{proof}
    The two items are clearly mutually exclusive, so it suffices to show that the negation of the first implies the second.
    Suppose the first item fails.
    Then there exists a language $K$ such that $\calU \in \uP^K$.
    We will construct an $\FNEXP^{L[\poly]}$ relation whose witness oracles allow a polynomial-time unitary algorithm to implement $\calU$.

    Let $\{C_{n,\eps}\}_{n \in \N,\eps>0}$ be a time-uniform $K$-oracle unitary circuit family witnessing $\calU \in \uP^K$.
    Fix $n$ and $\eps>0$, and set
    \[
    \delta \coloneqq \frac{\eps}{6\sqrt{2^n}}.
    \]
    We now define an $\FNEXP^{L[\poly]}$ relation $R$.
    The input to $R$ is $z = (1^n, \eps)$.
    A witness $w$ for $z$ is interpreted as a truth table for $K$ on all strings of length at most $q(n,\log 1/\eps)$, where $q$ is a sufficiently large polynomial upper bound on the lengths of all $K$-queries appearing in $C_{n, \delta}$.
    Let $C^w_{n,\delta}$ denote the unitary circuit obtained from $C_{n, \delta}$ by answering each $K$-query according to the table $w$.

    The verifier for $R$ will check that $C^w_{n,\delta}$ approximately implements $U_n$ on the zero-ancilla subspace.
    More precisely, let $n+a$ be the number of qubits on which $C_{n,\delta}$ acts, and set
    \[
    \gamma \coloneqq \frac{\delta}{2\sqrt{2^{n+a}}}.
    \]
    For every $x \in \{0,1\}^{n+a}$ and $y \in \{0,1\}^n$, it first computes fixed-point binary estimates $a_{xy}$ and $b_{xy}$ satisfying
    \begin{equation}
    \abs{a_{xy} - \bra{x}C^w_{n,\delta}\ket{y}\ket{0^a}} \le \gamma,
    \qquad
    \abs{b_{xy} - \bra{x}(U_n\ket{y})\ket{0^a}} \le \gamma.
    \label{eq:axy_bxy_delta}
    \end{equation}
    It accepts $w$ if for every $y \in \{0,1\}^n$,
    \[
    \sum_{x \in \{0,1\}^{n+a}} \abs{a_{xy} - b_{xy}}^2 \le 4\delta^2.
    \]

    Note that the estimates $a_{xy}$ can be computed in time $2^{\poly(n, \log 1/\eps)} \le 2^{\poly(|z|)}$ by sequential matrix multiplication of the gates comprising the circuit $C^w_{n,\delta}$, so long as an FPTEAS is used to evaluate entries of each non-oracle gate to sufficiently small error.
    Also, the estimates $b_{xy}$ are computable in space $\poly(n, \log 1/\eps)$ with the $L$-oracle FPSEAS for $\calU\Entry$, which exists by \Cref{thm:entrywise_characterization} (and when the last $a$ bits of $x$ are not all zero, we can simply take $b_{xy}=0$).
    Hence, each $b_{xy}$ is computable in time $2^{\poly(|z|)}$ with $\poly(|z|)$-length queries to $L$.
    So, the relation $R$ is certainly in $\FNEXP^{L[\poly]}$.

    Now let $W$ be any witness oracle for $R$.
    We construct a time-uniform $W$-oracle unitary circuit family for $\calU$.
    On input $(1^n,\eps)$, the uniformity machine first constructs the circuit $C_{n, \delta}$.
    It then replaces each query to $K$ on a string $s$ by a query to $W$ of the form
    \[
    (1^n,\eps, i_s),
    \]
    where $i_s$ is the index of $s$ in the truth-table ordering used for witnesses of $R$.
    This index is computable by a polynomial-time reversible classical computation, since $|s|\le \poly(n,\log 1/\eps)$.

    To prove that this circuit family witnesses $\calU \in \uP^R$, we first observe that the relation $R$ has at least one accepting witness for every $z=(1^n,\eps)$.
    Namely, take $w$ to be the restriction of the original oracle $K$ to all strings of length at most $q(n,\log 1/\eps)$.
    Then $C^w_{n,\delta} = C_{n,\delta}$, so by the assumption that $C_{n,\delta}$ implements $U_n$ to operator norm error $\delta$, we have, for every $y \in \{0,1\}^n$,
    \[
    \norm{C^w_{n,\delta}\ket{y}\ket{0^a} - (U_n\ket{y})\ket{0^a}} \le \delta.
    \]
    Thus by \Cref{eq:axy_bxy_delta} and the triangle inequality for Euclidean norm, we certainly have
    \[
    \left(\sum_{x \in \{0,1\}^{n+a}} \abs{a_{xy} - b_{xy}}^2\right)^{1/2}
    \le \delta + 2\gamma\sqrt{2^{n+a}} = 2\delta,
    \]
    and thus $w$ is accepted.

    Since $W$ is a witness oracle for $R$, queries of the form $(z, \cdot)$ to $W$ report the bits of a single witness $w$ for which $(z, w) \in R$.
    Therefore the resulting circuit is identical to $C_{n,\delta}^w$.
    By the definition of $R$, we know that for every $y \in \{0,1\}^n$,
    \[
    \left(\sum_{x \in \{0,1\}^{n+a}} \abs{a_{xy} - b_{xy}}^2\right)^{1/2} \le 2\delta,
    \]
    and therefore, by \Cref{eq:axy_bxy_delta} and the triangle inequality for Euclidean norm,
    \[
    \norm{C^w_{n,\delta}\ket{y}\ket{0^a} - (U_n\ket{y})\ket{0^a}}
    \le 2\delta + 2\gamma\sqrt{2^{n+a}} = 3\delta.
    \]
    Thus, by \Cref{lem:basis_to_op_norm}, it follows that for every $n$-qubit state $\ket{\psi}$,
    \[
    \norm{
        C^w_{n,\delta}\ket{\psi}\ket{0^a}
        -
        (U_n\ket{\psi})\ket{0^a}
    }
    \le
    3\delta\sqrt{2^n} = \eps/2.
    \]
    By \Cref{lem:extension_close_op_norm}, there is an extension $\overline{U}_n$ of $U_n$ such that
    \[
    \opnorm{C^w_{n,\delta}-\overline{U}_n}
    \le 6\delta\sqrt{2^n} = \eps,
    \]
    which establishes that $\calU \in \uP^W$.
\end{proof}

The following is immediate.

\begin{corollary}
    \label{cor:u-all-up-all}
    If $\uALL = \uP^\ALL$, then for any language $L$, $\uPSPACE^L \subseteq \uP^{\FNEXP^{L[\poly]}}$.
\end{corollary}

Now, we move on to exploring the classical complexity collapses that would allow for more efficient computation of $\NEXP$ witness oracles.
This next lemma is noteworthy because it says that a collapse of complexity classes with \textit{no} oracles would imply a space-efficient algorithm for $\NEXP$ witness finding relative to \textit{every} oracle.

\begin{lemma}
    \label{lem:fnp_fnc}
    Suppose $\NC = \NP$.
    Then for every language $L$, every relation $R \in \FNEXP^{L[\poly]}$ has a witness oracle in $\PSPACE^L$.
\end{lemma}

\begin{proof}
    Let $M^L(x,y)$ decide membership in $R$ in time $2^{\poly(|x|)}$.
    Let $p(|x|)$ be a sufficiently large polynomial such that $M^L(x,y)$ makes queries to $L$ of length at most $p(|x|)$, and for which the runtime of $M^L(x, y)$ and the witness length $|y|$ are at most $2^{p(|x|)}$.

    Define a language $R' \in \NP$ as follows.
    An input to $R'$ consists of $z=(x,A)$,
    where $A$ is interpreted as a truth table on all strings of length at most $p(|x|)$.
    A witness $y$ for $z$ is accepted if $M(x,y)$ accepts when its oracle queries are answered according to $A$.
    The language $R'$ is in $\NP$ because the simulation of $M$ and all lookups in $A$ take time polynomial in the input length, which is $|z| = |x| + 2^{p(|x|) + 1} - 1$.

    The assumption $\NP \subseteq \NC$ implies that $\mathsf{NC} = \P = \NP = \P^\NP$, and therefore $\mathsf{NC}$ circuits can \textit{find} $\mathsf{NP}$ witnesses (when they exist) by invoking the standard polynomial-time search-to-decision reduction for any $\NP$-complete problem.
    Hence, there is a logspace-uniform Boolean circuit family $\{C_n\}_{n \in \N}$ of size $\poly(n)$ and depth $\polylog(n)$ such that, whenever $z \in \{0,1\}^n$ is a yes-instance of $R'$, the output $C(z)$ is a valid witness $y$.
    We now define the witness oracle $W$ by feeding $L$ into $C$.
    Concretely, for each $x$, let $A_x$ be the truth table of $L$ on all strings of length at most $p(|x|)$, and define
    \[
    W(x,i) \coloneqq C(x, A_x)_i.
    \]
    It is clear that $W$ is a witness oracle for $R$, because if $x$ has a valid witness then $y = C(x, A_x)$ is one such witness.

    To conclude the proof, we show that $W\in\PSPACE^L$.
    To compute an output bit of $C(x, A_x)$, recursively evaluate the circuit from the corresponding output gate.
    Logspace uniformity of the circuit family allows the predecessors and type of each gate to be computed in space $O(\log(2^{p(|x|)})) \le \poly(|x|)$.
    Additionally, input bits belonging to $x$ are immediate, and input bits belonging to $A_x$ can be computed by a single query to $L$ on a string of length at most $p(|x|)$.
    As the recursion depth is $\polylog(2^{p(|x|)}) \le \poly(|x|)$, an output bit $C(x, A_x)_i$ can be evaluated in $\poly(|x|)$ space given $x$ and $i$, which puts $W\in\PSPACE^L$.
\end{proof}

The next lemma is formally incomparable to the last, but gives a different sufficient condition under which $\NEXP$ has more efficient witness oracles.

\begin{lemma}[{Easy witness lemma~\cite[Theorem 29]{IKW02-easy-witness}}]
    \label{lem:ikw_easy_witness}
    Suppose $L$ is a language for which $\AM^L = \NEXP^L$.
    Then every relation $R \in \FNEXP^{L}$ has a witness oracle in $\EXP^L$.
\end{lemma}

Of course, the $\EXP^L$ in the last lemma can be replaced by any complexity class between $\AM^L$ and $\NEXP^L$, including (say) $\PSPACE^L$.

\begin{remark}
    It is tempting to conclude that under the weaker assumption $\EXP = \NEXP$, $\NEXP$ witnesses should be findable in exponential time, e.g.\ by scaling up the standard polynomial-time search-to-decision reduction for $\NP$.
    This argument fails because the exponential-time $\NEXP$ decision algorithm would eventually have to be queried on strings that are \textit{themselves} exponentially-long, therefore taking doubly exponential time.
    More succinctly, the difficulty can be summarized as: $\P$ is self-low, but $\EXP$ is not.
    Thus, \Cref{lem:ikw_easy_witness} is substantially less trivial than it might seem.
\end{remark}

Combining the last two lemmas, we obtain:

\begin{corollary}
    \label{cor:nc_or_am}
    Let $L$ be any language, and suppose either $\NC = \NP$ or $\AM^L = \NEXP^L$.
    Then every relation $R \in \FNEXP^{L[\poly]}$ has a witness oracle in $\PSPACE^L$.
\end{corollary}

Further combining with our previous corollary on unitary synthesis, we see that \textit{even establishing an oracle separation} between $\uP^{\PSPACE}$ and $\uPSPACE$ would yield breakthrough lower bounds.

\begin{corollary}
    \label{cor:nc_np_separation}
    If there is a language $L$ relative to which $\uP^{\PSPACE^L} \neq \uPSPACE^{L}$, then at least one of the following holds:
    \begin{enumerate}
        \item $\uP^{\ALL} \neq \uALL$, or
        \item $\NC \neq \NP$ and $\AM^L \neq \NEXP^L$.
    \end{enumerate}
\end{corollary}

\begin{proof}
    We prove the contrapositive.
    Suppose both items fail, in which case $\uALL = \uP^\ALL$, and either $\NC = \NP$ or $\AM^L = \NEXP^L$.
    The first assumption implies $\uPSPACE^L \subseteq \uP^{\FNEXP^{L[\poly]}}$, by \Cref{cor:u-all-up-all}.
    The second assumption implies that every $\FNEXP^{L[\poly]}$ relation has a witness oracle in $\PSPACE^L$, by \Cref{cor:nc_or_am}.
    We conclude that
    \[
    \uPSPACE^L \subseteq \uP^{\FNEXP^{L[\poly]}} = \uP^{\PSPACE^L},
    \]
    which holds for any language $L$.
\end{proof}

Weakening the assumption still yields a powerful consequence:

\begin{corollary}
    \label{cor:bpp_nexp}
    If there is a language $L$ relative to which $\uP^{L} \neq \uPSPACE^{L}$, then at least one of the following holds:
    \begin{enumerate}
        \item $\uP^{\ALL} \neq \uALL$, or
        \item $\BPP^L \neq \NEXP^L$.
    \end{enumerate}
\end{corollary}

\begin{proof}
    We prove the contrapositive.
    Suppose both items fail, in which case $\uALL = \uP^\ALL$ and $\BPP^L = \AM^L = \PSPACE^L = \NEXP^L$.
    Thus, the contrapositive of \Cref{cor:nc_np_separation} tells us that $\uPSPACE^L = \uP^{\PSPACE^L} = \uP^{\BPP^L}$.
    The latter is equal to $\uP^L$ by \Cref{prop:tidy_contains_bbqp,prop:tidy_low}.
\end{proof}

\begin{remark}
    The conditional consequences $\AM^L \neq \NEXP^L$ and $\BPP^L \neq \NEXP^L$ in the above corollaries only represent breakthrough separations when $L$ is the empty oracle, because there \textit{do} exist languages $L$ relative to which $\BPP^L = \AM^L \neq \NEXP^L$ (for example, any $\PSPACE$-complete language).
\end{remark}

\subsection{Implications for cryptography}

We conclude by investigating the cryptographic consequences of these separations, with the particular goal of understanding whether the existence of any form of quantum computational cryptography is subject to a complexity-theoretic barrier.
Specifically, we consider EFI pairs, which are equivalent to quantum commitment schemes and are implied by the vast majority of computationally-secure quantum cryptography~\cite{Yan22-commitments,BCQ23-efi}.

\begin{definition}[EFI pair]
    An \emph{EFI pair} is a pair of time-uniform quantum circuit families $\{C_{0,n}\}_{n \in \N}$ and $\{C_{1,n}\}_{n \in \N}$ whose outputs on security parameter $n$
    are mixed states $\rho_{0,n}$ and $\rho_{1,n}$, respectively, such
    that the following hold:
    \begin{enumerate}
        \item There exists a polynomial $p$ such that, for every
        sufficiently large $n$,
        \[
            \frac{1}{2}
            \norm{\rho_{0,n}-\rho_{1,n}}_1
            \geq
            \frac{1}{p(n)}.
        \]
        \item For every polynomial-time quantum circuit family $\{D_n\}_{n \in \N}$,\footnote{The circuit family need not be uniform, although this distinction will be irrelevant to our results.}
        \[
            \abs{
                \Pr[D_n(\rho_{0,n})=1]
                -
                \Pr[D_n(\rho_{1,n})=1]
            }
            \le \negl(n).
        \]
    \end{enumerate}
    For a language $L$, an \textit{EFI pair relative to $L$} is defined identically, but with $\{C_{0,n}\}_{n \in \N}$, $\{C_{1,n}\}_{n \in \N}$, and $\{D_n\}_{n \in \N}$ all allowed to be $L$-oracle circuit families.
\end{definition}

The key fact we need is that EFI pairs do not exist if $\uP = \uPSPACE$.
A version of this was derived by~\cite{BEMPQY26-uhlmann}, but for a different definition of unitary complexity classes and a more general set of cryptographic primitives.

\begin{theorem}
    \label{prop:efi_separates_up_upspace}
    If EFI pairs exist relative to a language $L$, then $\uP^L \neq \uPSPACE^L$.
\end{theorem}

\begin{proof}
    Let $\{\rho_{0,n}, \rho_{1,n}\}_{n \in \N}$ be mixed states produced by an EFI pair relative to $L$, and assume without loss of generality that they are $n$-qubit states.
    Define $A_n = \rho_{0,n} - \rho_{1,n}$, and let $\delta \ge 1/\poly(n)$ be a lower bound on the trace distance $\frac{1}{2}\norm{A_n}_1$ between $\rho_{0,n}$ and $\rho_{1,n}$.
    
    We first observe that the entries $\braket{x}{A_n|y}$ of $A_n$ admit an FPSEAS relative to $L$, by a standard Feynman path integral calculation such as the one we performed in \Cref{prop:bqpspace}.
    Consequently, we claim that this yields an FPSEAS relative to $L$ for the entries of the unitary $U_n = e^{i A_n}$.
    Indeed, it suffices to truncate the Taylor series
    \[
    \sum_{k=0}^\infty \frac{i^k A_n^k}{k!}
    \]
    at degree $k = O(\log 1/\eps)$, as $\opnorm{A_n} \le 1$:
    \[
    \opnorm[\Big]{e^{i A_n} - \sum_{k=0}^K \frac{i^k A_n^k}{k!}}
    \le
    \sum_{k=K+1}^\infty \frac{\opnorm{A_n}^k}{k!}
    \le \sum_{k=K+1}^\infty \frac{1}{k!}
    \le \frac{e}{(K+1)!}.
    \]
    Each entry of the truncated series can be evaluated in space polynomial in $n$ and $\log 1 / \eps$, by expanding the matrix products and summing the terms with sufficient precision.
    
    As a consequence of the entrywise characterization of \Cref{thm:entrywise_characterization}, we deduce that the unitary language $\calU = \{U_n\}$ is contained in $\uPSPACE^L$.

    Choose the parameters
    \begin{equation}
    \label{eq:def_tau_eta}
    \tau \coloneqq \frac{\delta}{2^{n+2}}, \qquad \eta \coloneqq \frac{\delta}{8}.
    \end{equation}
    Let $V_n$ be the unitary phase estimation circuit for $U_n$ that, on input an eigenstate with eigenvalue $e^{i \lambda}$, produces an estimate of $\lambda \in [-1, 1]$ to additive error $\tau$, except with failure probability at most $\eta$.\footnote{Note that $V_n$ acts on more than $n$ qubits.}
    Then $\calV = \{V_n\}_{n \in \N}$ is contained in $\uPSPACE^\calU$, because phase estimation requires $O(n + \log 1/\tau + \log 1/\eta) \le \poly(n)$ qubits and $2^{\poly(n)}$ controlled calls to $U_n$.
    By the self-lowness of $\uPSPACE$ (\Cref{prop:tidy_low}), we conclude that $\calV \in \uPSPACE^L$ as well.

    Now consider the distinguisher $D_n(\rho)$ that applies $V_n$ to the input state $\rho$ and accepts when the estimated eigenvalue is nonnegative.
    Decompose $A_n$ into an eigenbasis
    \[
    A_n = \sum_{j=1}^{2^n} \lambda_j \ketbra{v_j}{v_j},
    \]
    and let $q_j$ be the acceptance probability of $D(\ket{v_j})$.
    Then the advantage satisfies
    \begin{align*}
        \Pr[D(\rho_{0,n}) = 1] - \Pr[D(\rho_{1,n}) = 1]
        &= \sum_{j=1}^{2^n} \lambda_j q_j\\
        &\ge \sum_{\lambda_j > \tau} \lambda_j (1 - \eta) + \sum_{0 < \lambda_j \le \tau} 0 + 
        \sum_{-\tau \le \lambda_j \le 0} \lambda_j + \sum_{\lambda_j < -\tau} \lambda_j \eta\\
        &= \sum_{\lambda_j > 0} \lambda_j - \sum_{\abs{\lambda_j} \le \tau} \abs{\lambda_j} - \eta \sum_{\abs{\lambda_j} > \tau} \abs{\lambda_j}\\
        &\ge \sum_{\lambda_j > 0} \lambda_j - \sum_{j = 1}^{2^n} \tau - \eta \sum_{j = 1}^{2^n}\abs{\lambda_j}\\
        &\ge \delta - 2^n\tau - 2\eta\\
        &= \delta / 2.
    \end{align*}
    Above, the second line uses the fact that $q_j \ge (1 - \eta)$ when $\lambda_j > \tau$, $q_j \le \eta$ when $\lambda_j < -\tau$, and otherwise trivially $q_j \in [0,1]$; the fourth line expands the range of the subtracted summations; the fifth line uses the fact that $\norm{A_n}_1 \le \norm{\rho_{0,n}}_1 + \norm{\rho_{1,n}}_1 = 2$; and the last line substitutes the definitions of $\tau$ and $\eta$ (\Cref{eq:def_tau_eta}).

    To conclude, suppose for a contradiction that $\uP^L = \uPSPACE^L$.
    Then we could implement $V_n$ to error $\delta / 16$ in $\poly(n)$ time.
    Using this implementation in place of $V_n$ in $D$ would maintain a $\poly(n)$ runtime and change the acceptance probability of $D$ by at most $\delta / 8$, leaving distinguishing advantage at least $\delta / 4$, which is non-negligible.
    That would violate the assumption that the states define a valid EFI pair, so we must have $\uP^L \neq \uPSPACE^L$.
\end{proof}

Plugging into the corollaries from before directly yields the following:

\begin{corollary}
    Suppose $L$ is a language for which an EFI pair exists relative to $\PSPACE^L$.
    Then at least one of the following holds:
    \begin{enumerate}
        \item $\uALL \neq \uP^\ALL$, or
        \item $\NC \neq \NP$ and $\AM^L \neq \NEXP^L$.
    \end{enumerate}
\end{corollary}

\begin{proof}
    By \Cref{prop:efi_separates_up_upspace}, the assumption of the corollary implies that $\uP^{\PSPACE^L} \neq \uPSPACE^{\PSPACE^L}$, and therefore $\uP^{\PSPACE^L} \neq \uPSPACE^L$ by the self-lowness of $\uPSPACE$ (\Cref{prop:tidy_low}).
    The claim thus follows from \Cref{cor:nc_np_separation}.
\end{proof}

\begin{corollary}
    \label{cor:efi_bpp_nexp}
    Suppose $L$ is a language for which an EFI pair exists relative to $L$.
    Then at least one of the following holds:
    \begin{enumerate}
        \item $\uALL \neq \uP^\ALL$, or
        \item $\BPP^L \neq \NEXP^L$.
    \end{enumerate}
\end{corollary}

\begin{proof}
    By \Cref{prop:efi_separates_up_upspace}, the assumption of the corollary implies that $\uP^{L} \neq \uPSPACE^{L}$.
    The claim thus follows from \Cref{cor:bpp_nexp}.
\end{proof}

\section{AI Disclosure}

We used ChatGPT to help prepare the manuscript.
GPT suggested some of the notation, including for uniform circuit families and the name $\puP$.
GPT's largest contributions were to the proofs in \Cref{sec:characterization,sec:synthesis}.
We wrote the proofs of \Cref{lem:fnp_fnc,lem:rotation_to_fqpseas,thm:projunitarypspace_equiv,thm:no_global_phase,lem:unitarypspace_dichotomy_witnexp,lem:u_to_entry,prop:efi_separates_up_upspace,lem:entry_to_column} interactively with GPT by providing the model with sketches conceived by us.
In the case of \Cref{thm:projunitarypspace_equiv}, GPT considerably simplified our proof strategy.
For \Cref{lem:u_to_entry,prop:efi_separates_up_upspace,lem:entry_to_column}, GPT suggested substantially different proofs; we chose to include them because they are simpler.
GPT provided the proof of \Cref{lem:no_approximate_projective_section} en route to \Cref{thm:no_global_phase}, with minimal input from the authors.
The authors verified all of the AI-generated proofs and reworked them for clarity.
We used ChatGPT-5.5 everywhere, with the sole exception of \Cref{prop:efi_separates_up_upspace}, which used ChatGPT-6.

\section{Acknowledgments}
We thank Fermi Ma and John Wright for vigorously debating the merits of our definitions in the early stages of this work.
Part of this work was done while W.K.\ was visiting the Simons Institute for the Theory of Computing during the 2025 Quantum Industry Day
and the 2026 Summer Cluster on Quantum Computing.
E.T.\ is supported by the Miller Institute for Basic Research in Science, University of California Berkeley.

\printbibliography

\end{document}